\documentclass[10pt]{article}

\usepackage{geometry}
\usepackage{lmodern}
\usepackage[T1]{fontenc}
\usepackage[mathletters]{ucs} % Extended unicode (utf-8) support
\usepackage[utf8]{inputenc}

\usepackage{graphicx}
\graphicspath{{images/}}

\usepackage{amsmath, amsthm}
\usepackage{amssymb, amsbsy}
\usepackage[dvipsnames]{xcolor}
\usepackage{hyperref}
\hypersetup{
colorlinks=true,
linkcolor=Brown,
citecolor=OliveGreen,
urlcolor=RoyalBlue,
}

\usepackage{floatflt, wrapfig}
\usepackage{xcolor}

\usepackage{graphicx}
\graphicspath{ {images/} }
\usepackage{epstopdf}
\usepackage{color, import}
\usepackage{sectsty}
\usepackage{enumitem}
\usepackage{lineno}

\usepackage{cite}
\usepackage[edges]{forest}

\usepackage[linesnumbered,vlined,ruled]{algorithm2e}
\SetAlFnt{\footnotesize}
\SetAlCapFnt{\small}
\SetAlCapNameFnt{\small}

\usepackage{setspace}
\usepackage[list=true]{subcaption}
\usepackage[title]{appendix}
\usepackage{multirow}
\usepackage{array}

\usepackage{etoolbox}
\makeatletter
\patchcmd{\@maketitle}{\begin{center}}{\begin{flushleft}}{}{}
\patchcmd{\@maketitle}{\begin{tabular}[t]{c}}{\begin{tabular}[t]{@{}l}}{}{}
\patchcmd{\@maketitle}{\end{center}}{\end{flushleft}}{}{}
\makeatother

\let\oldnl\nl% Store \nl in \oldnl
\newcommand{\nonl}{\renewcommand{\nl}{\let\nl\oldnl}}

\renewenvironment{abstract}
{\small\section*
{\bfseries\noindent{\raisebox{-.15\baselineskip}{\normalsize\abstractname}}\hrulefill} 
}

\newtheorem{theorem}{Theorem}
\newtheorem{lemma}{Lemma}
\newtheorem{corollary}{Corollary}

\author{hkhj\\hjk}
\author{hjjkhj\\hjk}

\begin{document}
 \pagenumbering{gobble}

\title{Distributed Localization of Wireless Sensor Network Using Communication Wheel}
    \author{    
    Kaustav Bose\\
    \small \emph{Department of Mathematics, Jadavpur University, India}\\
    \small \emph{kaustavbose.rs@jadavpuruniversity.in}\\\\
    Manash Kumar Kundu\\
    \small \emph{Gayeshpur Government Polytechnic, Kalyani, India}\\
    \small \emph{manashkrkundu.rs@jadavpuruniversity.in}\\\\
    Ranendu Adhikary\\
    \small \emph{Department of Mathematics, Jadavpur University, India}\\
    \small \emph{ranenduadhikary.rs@jadavpuruniversity.in}\\\\
    Buddhadeb Sau\\
    \small \emph{Department of Mathematics, Jadavpur University, India}\\
    \small \emph{buddhadeb.sau@jadavpuruniversity.in}\\\\
    }
    \date{}
                  % typeset the header of the contribution
    
  \maketitle

\begin{abstract}
% 
% We show that the \textsc{Max Cut} problem is NP-complete on interval graphs. This resolves the longstanding open problem concerning the computational complexity \textsc{Max Cut} on interval graphs.

We study the \emph{network localization problem}, i.e., the problem of determining node positions of a wireless sensor network modeled as a unit disk graph.  In an arbitrarily deployed network, positions of all nodes of the network may not be uniquely determined. It is known that even if the network corresponds to a unique solution, no polynomial-time algorithm can solve this problem in the worst case, unless RP = NP. So we are interested in algorithms that efficiently localize the network partially. A widely used technique that can efficiently localize a uniquely localizable portion of the network is \emph{trilateration}: starting from three \emph{anchors} (nodes with known positions), nodes having at least three localized neighbors are sequentially localized. However, the performance of trilateration can substantially differ for different choices of the initial three anchors. In this paper, we propose a distributed localization scheme with a theoretical characterization of nodes that are guaranteed to be localized. In particular, our proposed distributed algorithm starts localization from a \emph{strongly interior node} and provided that the subgraph induced by the strongly interior nodes is connected, it  localizes all nodes of the network except some \emph{boundary nodes} and \emph{isolated weakly interior nodes}.
\end{abstract}
\section{Introduction}

%  \paragraph{Sensor Network Localization.}
 
 A \emph{wireless sensor network} (WSN) is a wireless network consisting of a large number of small autonomous sensors spatially distributed in a region to monitor physical or environmental parameters. The sensor nodes are low-cost, low-power, autonomous, multi-functional devices equipped with sensing, processing, and communication capabilities. The knowledge of the physical location of sensor nodes is essential in many applications where the geographical information of the sensed data is important, for example, event detection, environment and habitat monitoring, target tracking, pervasive medical care, etc. The positional information of the nodes also supports many fundamental location-aware protocols, like geographic routing, topology control, coverage, etc. One method of determining the location of the nodes is by equipping the sensor nodes with Global Positioning System (GPS). However, the installation of GPS on each node of a large scale WSN is expensive and  the power consumption of GPS reduces the battery life of the sensor nodes. Moreover, it is not suitable in dense forests, underground or indoor environment where GPS signals are unavailable. Therefore, novel schemes have been proposed to determine the positions of the nodes in a network where only some special nodes called \emph{anchors} are aware of their positions with respect to some global coordinate system (e.g., \cite{savvides2001dynamic,moore2004robust,bulusu2000gps,albowicz2001recursive,biswas2006semidefinite,biswas2008distributed}). In these schemes, the nodes can measure the distances to their neighboring nodes and using these distance information they try to determine their positions. This process of computing the positions of the nodes is called \emph{range-based network localization} or simply \emph{network localization}.

 The network localization problem can be abstracted as the following: given a weighted graph with edge weights equal to the distances between the respective nodes and coordinates of some nodes, called anchors, with respect to some coordinate system, we have to compute the coordinates of all other nodes in that coordinate system. A network, with the given positions of anchors and distances between adjacent nodes, is said to be \emph{uniquely localizable} if all nodes of the network have unique positions consistent with the given data, i.e., there is a unique solution. Obviously, if the given instance corresponds to multiple feasible solutions, the actual positions of the nodes can not be determined. The unique localizability of a network is completely determined by certain combinatorial properties of the network graph and the number of anchors. \emph{Graph rigidity theory} \cite{eren2004rigidity,hendrickson92,jackson2005connected} provides the following necessary and sufficient condition for unique localizability \cite{eren2004rigidity}: a network is uniquely localizable if and only if it has at least 3 anchors and the network graph is \emph{globally rigid} (See Section \ref{sec: rigid} for definition). However, unless a network is highly dense and regular, it is unlikely that the network is globally rigid. But even if a network is not globally rigid as a whole, a large portion of the network may be globally rigid. For the remaining nodes, there are multiple feasible solutions and hence, their actual positions can not be determined. In the decision version of the problem, also known as \textsc{Graph Embedding} or \textsc{Graph Realization} problem, given a weighted graph we have to determine whether there is an embedding of the graph in Euclidean plane so that the distances between the adjacent vertices are equal to the edge weights. This problem has been shown to be strongly NP-hard \cite{saxe1979}. In \cite{eren2004rigidity}, it is shown that the problem remains NP-hard even when the graph is globally rigid. However, these results are for general graphs. In a sensor network, only nodes that are within a certain communication range, say $r$, can measure their relative distances. Therefore, the network can be better modeled as a unit disk graph: two nodes are adjacent if and only if their distance is $\leq r$. In this version of the problem, apart from the coordinates of the anchors and the distances between the adjacent nodes, we have a third type of information: the distances between the non-adjacent nodes are $> r$. The decision version of this problem, also known as \textsc{Unit Disk Graph Reconstruction} problem, is that given a weighted graph with weights $\leq r$, we have to determine whether there is an embedding of the graph in Euclidean plane so that 1) the distances between the adjacent vertices are equal to the edge weights, and 2) the distance between any pair of non-adjacent nodes is $> r$. It is shown in \cite{aspnes2004} that \textsc{Unit Disk Graph Reconstruction} is NP-hard. Therefore, there is no efficient algorithm that solves the localization problem in the worst case unless P = NP. It is further shown in \cite{aspnes2004} that a similar result holds even for instances that have unique reconstructions: there is no efficient randomized algorithm that solves the localization problem even for instances that have unique reconstructions unless RP = NP.

Since a real life instance may not have unique solution and even if it has, it is unlikely that there is an efficient algorithm that solves the  problem, we are interested in efficient heuristics that partially localize the network. A very popular technique is  \emph{trilateration} which efficiently localizes a globally rigid subgraph of the network. It is based on the simple fact that the position of a node can be determined from its distance from three non-collinear nodes with known coordinates. The algorithm starts with at least three anchor nodes and then nodes adjacent to at least three nodes with known coordinates are sequentially localized. It is computationally efficient and very easy to implement in distributed setting, thus widely used in practice. In this paper, we are interested in \emph{anchor-free localization}, i.e., there are no anchor nodes. Since for localization at least three anchor nodes are necessary, in the anchor-free case, some three mutually adjacent nodes of the network fix their coordinates (respecting their mutual distances) in some virtual coordinate system. These three nodes play the role of anchors. However, in case of trilateration, the performance of the algorithm can drastically differ for different choices of the initial three nodes. In this paper, we address this issue and propose a distributed anchor-free localization scheme with a theoretical characterization of nodes that are guaranteed to be localized. In our approach, a node, based on its local information, can categorize itself as either \emph{strongly interior}, \emph{non-isolated weakly interior}, \emph{isolated weakly interior} or \emph{boundary}. Provided that the \emph{strong interior}, i.e., the subgraph induced by the set of strongly interior nodes, is connected, one strongly interior node is chosen by a leader election protocol. Our sequential localization algorithm starts  from that strongly interior node, and it is theoretically guaranteed to localize all nodes except some boundary and isolated weakly interior nodes. Due to the space restrictions, it is not possible to present a comprehensive survey of the large number of works on localization (e.g.,  \cite{savvides2001dynamic,moore2004robust,bulusu2000gps,albowicz2001recursive,biswas2006semidefinite,biswas2008distributed,rong2006angle,lederer2009connectivity,goldenberg2006localization,yang2010beyond,fang2009sequential,he2003range,goldenberg2005network,shang2004improved,shang2003localization,ji2004sensor,baggio2008monte,sorbelli2018range} etc.) in the literature. The readers are instead referred to the surveys \cite{chowdhury2016advances,liu2010location,mao2007,wang2010survey} and the references therein.

\section{Preliminaries}

\subsection{Basic Model and Assumptions}

The mathematical model of wireless sensor network considered in this work is described in the following:

\begin{itemize}

 \item A set of $n$ sensors is arbitrarily deployed in $\mathbb{R}^2$. Each sensor node has computation and wireless
communication capabilities. 

\item There is a constant $r > 0$, called the \emph{communication range}, such that any two sensor nodes can directly communicate with each other if and only if the distance between them is $\leq r$. This implies that the corresponding communication network can be modeled as a \emph{unit disk graph (UDG)}: two nodes are adjacent if and only if they are at most $r$ distance apart. We assume that this graph is connected. Note that if the graph is not connected, then it is impossible to localize the entire network consistently. 

% \item There is another constant $s > 0$, called the \emph{communication range}, such that any sensor can observe or sense any point in the plane if and only if it is at a distance at most $s$ from the sensor. We assume that $r \leq r$. This assumption has been previously used in literature \cite{Zhang05,Xing06,Wang03} and is consistent with most commercially available sensor nodes.

\item The euclidean distance between a pair of sensors can be measured directly and accurately if and only if they are at most $r$ distance apart. Hence, if a sensor node can directly communicate with another node, then it also knows the distance between them. 

% 
% In practice, the distance between two sensors is measured based on time of flight (TOF) or received signal strength (RSS) etc. We shall assume that the distances measured by the sensors are noise free, i.e., measured with high degree of accuracy. 

\item The sensor nodes are assumed to be in general positions, i.e., no three points are collinear. This is not a major assumption, as the nodes of a randomly deployed network are almost always in general positions.

\end{itemize}

\subsection{Definitions and Notations}

Let $\mathcal{V}$ be the set of $n$ sensors at positions in $\mathbb{R}^2$. The corresponding wireless sensor network can be modeled as an undirected edge-weighted graph $\mathcal{G} = (\mathcal{V}, \mathcal{E}, w)$, where
\begin{itemize}

 \item $\mathcal{V} = \{ v_1, \ldots, v_n \}$ is the set of sensors,
 
 \item  $(v_i, v_j) \in \mathcal{E}$, i.e., $v_i$ is adjacent to $v_j$ if and only if  $d(v_i, v_j) \leq r$, where $r$ is the communication range of the sensors,
 
 \item the edge-weight $w: \mathcal{E} \longrightarrow \mathbb{R}$ is given by $w(v_i, v_j) = d(v_i, v_j)$.

\end{itemize}

We call $\mathcal{G}$ the \emph{underlying network graph} of the wireless sensor network. As mentioned previously, we assume that the graph $\mathcal{G}$ is connected.

%  \begin{floatingfigure}[r]{5cm}
%    
%    \fontsize{8pt}{8pt}\selectfont
%    \def\svgwidth{0.35\textwidth}
%    \import{}{images/eclipse.pdf_tex}
%    \caption{$u$ is not a maximal neighbor of $v$ as $u \preceq_v u'$.}
%    \label{eclipse_fig}
%    
% \end{floatingfigure}

\begin{figure}[htb!]
\centering
\fontsize{10pt}{10pt}\selectfont
\def\svgwidth{0.4\textwidth}
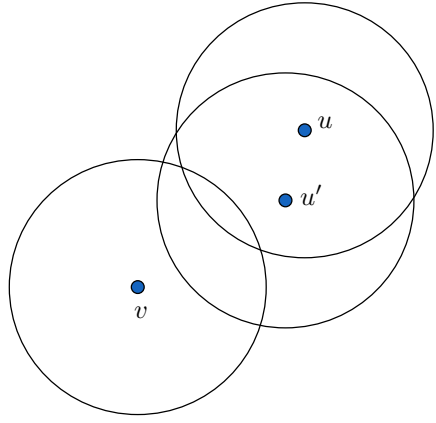
\caption{$u$ is not a maximal neighbor of $v$ as $u \preceq_v u'$.}
\label{eclipse_fig}
\end{figure}

A sensor node $v \in \mathcal{V}$ is called an \emph{interior node} if for every point $z \in \partial(\mathcal{Z}(v))$, where $\partial(\mathcal{Z}(v))$ is the boundary of $\mathcal{Z}(v)$, we have $z \in \mathcal{Z}(v')$ for some $v' \in \mathcal{V} \setminus \{v\}$. If $v \in \mathcal{V}$ is not an interior node, then it is called a \emph{boundary node}. An interior node $v \in \mathcal{V}$ is said to be \emph{a strongly interior node} if every node in $\mathcal{N}(v)$ is an interior node. An interior node $v \in \mathcal{V}$ is said to be \emph{a weakly interior node} if at least one node in $\mathcal{N}(v)$ is a boundary node. A weakly interior node is said to be \emph{isolated} if it is not adjacent to any strongly interior node.  The subgraph of $\mathcal{G}$ induced by the set of all interior nodes is called the \emph{interior} of $\mathcal{G}$. Similarly, the subgraph of $\mathcal{G}$ induced by the set of all strongly interior nodes is called the \emph{strong interior} of $\mathcal{G}$.

 If $v, v' \in \mathcal{V}$ are adjacent to each other, then we shall refer to the intersections of $\partial(\mathcal{Z}(v))$ and $\partial(\mathcal{Z}(v'))$ as their \emph{boundary intersections}. We shall denote these boundary intersections as $CW(v, v')$ and $CCW(v, v')$ according to the following rule: if one traverses from $CCW(v, v')$ to $CW(v, v')$ along $\partial(\mathcal{Z}(v))$ in clockwise direction, it sweeps an angle $< \pi$ about the center $v$.

Given a node $v$, we define a partial order relation $\preceq_v$ on $\mathcal{N}(v)$ as following: for $u, u' \in \mathcal{N}(v)$, $u \preceq_v u'$ if and only if $\mathcal{Z}(u) \cap \partial(\mathcal{Z}(v)) \subseteq \mathcal{Z}(u') \cap \partial(\mathcal{Z}(v))$. See Fig. \ref{eclipse_fig}. A node $u \in \mathcal{N}(v)$ is said to be a \emph{maximal} neighbor of $v$ if it is a maximal element in $\mathcal{N}(v)$ with respect to $\preceq_v$, i.e., there is no $u' \in \mathcal{N}(v) \setminus \{u\}$, such that $u \preceq_v u'$.

% We have the following results from elementary geometric arguments.

% 
% 
% \subsection{Problem statement}
% 

% \section{Localizability of the network}

\subsection{Some Results from Graph Rigidity Theory}\label{sec: rigid}

In this section, we present some basic definitions and results in graph rigidity. For a detailed exposition on graph rigidity, the readers are referred to \cite{jackson2005connected}.

A $d$-dimensional \emph{framework} is a pair $(G, \rho)$, where $G = (V, E)$ is a connected simple graph and the \emph{realization}  $\rho$ is a map $\rho : V \longrightarrow \mathbb{R}^d$. Two frameworks $(G, \rho_1)$ and $(G, \rho_2)$ are said to be \emph{equivalent} if $d(\rho_1(u),  \rho_1(v)) = d(\rho_2(u), \rho_2(v))$, for all $(u,v) \in E$. Frameworks $(G, \rho_1)$ and $(G, \rho_2)$ are said to be \emph{congruent} if $d(\rho_1(u), \rho_1(v)) = d(\rho_2(u), \rho_2(v))$, for all $u, v \in V$. In other words, two frameworks are said to be congruent if one can be obtained from another by an isometry of $\mathbb{R}^d$. A realization is \emph{generic} if the vertex coordinates are algebraically independent over rationals. The framework $(G, \rho)$ is \emph{rigid} if $\exists$ an $\varepsilon > 0$ such that if $(G, \rho')$
is equivalent to $(G, \rho)$ and $d(\rho(u), \rho'(u)) < \varepsilon$ for all $u \in V$, then $(G, \rho')$ is congruent
to $(G, \rho)$. Intuitively, it means that the framework can not be continuously deformed. $(G, \rho)$ is said to be \emph{globally rigid} if every framework which is equivalent to $(G, \rho)$ is congruent to $(G, \rho)$. It is known
\cite{whiteley1996some} that rigidity is a \emph{generic property}, that is, the rigidity of $(G, \rho)$ depends only
on the graph $G$, if $(G, \rho)$ is generic. The set of generic realizations is dense in the realization space and thus almost all realizations of a graph are generic.  So, we say that a graph $G$ is rigid in $\mathbb{R}^2$ if every generic realization of $G$ in $\mathbb{R}^2$ is rigid.

\begin{theorem}\cite{jackson2005connected}\label{jack}
  A graph $G$ is \emph{globally rigid} in $\mathbb{R}^2$ if and only if
either $G$ is a complete graph on at most three vertices or $G$ is 3-connected, rigid and remains rigid even after deleting an edge.
\end{theorem}

\begin{theorem}\cite{eren2004rigidity}\label{th: sufficient}
  If a network has at least 3 anchors and the underlying network graph is globally rigid, then it is uniquely localizable.
\end{theorem}

The condition of having at least 3 anchors is also necessary for unique localizability in order to rule out the trivial transformations. Since we are considering anchor-free localization, some three mutually adjacent nodes of the network will play the role anchors by fixing their coordinates (respecting their mutual distances) in some virtual coordinate system. The remaining nodes of the network have to find their  position according to this coordinate system. It should be noted here that for networks that do not satisfy the condition that two nodes are adjacent if and only if they are within some fixed distance, the condition of having globally rigid underlying network graph is also necessary. In our model, where two nodes are adjacent if and only if the distance between them is at most $r$, the network can be uniquely localizable even if its underlying network graph is not globally rigid.

\section{Construction of a Globally Rigid Subgraph Using Communication Wheels}\label{constr}

In this section, we shall show that if the strong interior is connected, then the network has a globally rigid subgraph containing all strongly interior nodes, and all non-isolated weakly interior nodes. The proof is constructive and will lead to our localization algorithm presented in section \ref{sec_algo}. 

% Due to space constraints, the proofs of the results stated in this section are provided in the Appendix.

 We first present some results that will be frequently used in the paper. Lemmas \ref{eclipse_distance_relation}-\ref{npreq} follow from elementary geometric arguments.

\begin{lemma}\label{perimeter_2cover}
 Let $v_1$ be an interior node and $v_2 \in \mathcal{N}(v_1)$. Then
 \begin{enumerate}
  \item $CCW(v_1, v_2) \in \mathcal{Z}(v_3)$ for some $v_3 \in \mathcal{N}(v_1) \setminus \{v_2\}$, such that $CCW(v_1, v_3) \notin \mathcal{Z}(v_2)$,
  
  \item $CW(v_1, v_2) \in \mathcal{Z}(v_4)$ for some $v_4 \in \mathcal{N}(v_1) \setminus \{v_2\}$, such that $CW(v_1, v_4) \notin \mathcal{Z}(v_2)$
 \end{enumerate}

\end{lemma}

\begin{proof}
 It is sufficient to prove only the first part. We shall prove by contradiction. So, assume that there is no such node in $\mathcal{N}(v_1) \setminus \{v_2\}$. Let $P = CCW(v_1,v_2)$. Let us partition the set of neighbors of $v_1$ into two sets as: $A = \{v \in \mathcal{N}(v_1) \mid P \in \mathcal{Z}(v) \}$ and $B = \{v \in \mathcal{N}(v_1) \mid P \notin \mathcal{Z}(v) \}$. $A \neq \emptyset$, since $v_2 \in A$. $B \neq \emptyset$, because the diametrically opposite point of $P$ on $\partial(\mathcal{Z}(v_1))$ must be covered by some node which does not cover $P$.
 
  Fix the ray $\overrightarrow{v_1 P}$ as a reference axis. Now for each $v \in \mathcal{N}(v_1)$, shoot rays from $v_1$ passing through $CCW(v_1, v)$ for $v \in A$ and $CW(v_1, v)$ for $v \in B$. For each $v \in \mathcal{N}(v_1)$, let $\theta_v$ be the angle formed by the corresponding ray measured counterclockwise from the reference axis $\overrightarrow{v_1 P}$. Let $\theta = min\{\theta_v \mid v \in B\}$. We must have $\theta > 0$, since for any $v \in B$, $\theta_v > 0$. Also, it implies from our hypothesis that $max\{\theta_v \mid v \in A\} = 0$. Then clearly any point on $\partial(\mathcal{Z}(v_1))$ making an angle in between $(0, \theta)$ with the ray $\overrightarrow{v_1 P}$ is not covered by any neighbor of $v_1$ (See Fig. \ref{2cov_fig}). This contradicts the fact that $v_1$ is an interior node. 
\end{proof}

\begin{figure}[htb!]
\centering
\fontsize{10pt}{10pt}\selectfont
\def\svgwidth{0.5\textwidth}
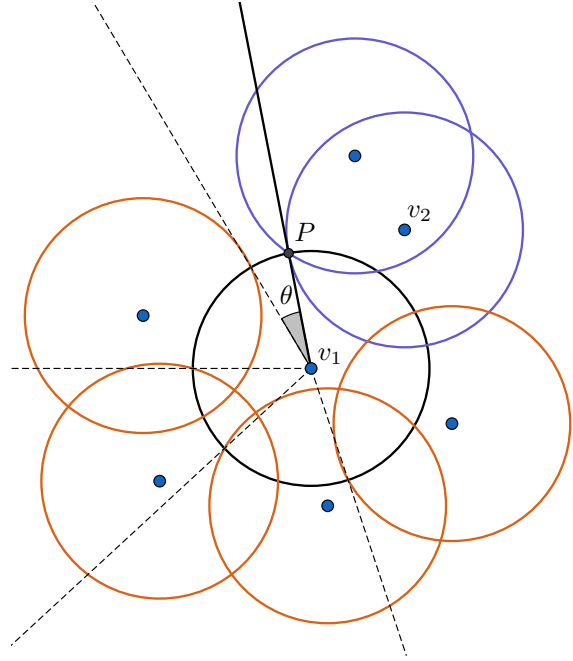
\caption{Illustration of the constructions in the proof of Lemma \ref{perimeter_2cover}. The purple and the orange circles correspond to the boundaries of the communication zones of the nodes in the set $A$ and $B$ respectively.}
\label{2cov_fig}
\end{figure}

\begin{lemma}\label{eclipse_distance_relation}
 If $u$ and $u'$ are two distinct neighbors of $v \in \mathcal{V}$ such that $u \preceq_v u'$, then $d(u, v) > d(u', v)$.
\end{lemma}

\begin{lemma}\label{eclipse_dual_relation}
 For distinct $v, u, u' \in \mathcal{V}$, $u \preceq_v u' \Leftrightarrow v \preceq_u u'$ .
\end{lemma}

\begin{lemma}\label{maximal_dual_relation}
 For distinct $v, u \in \mathcal{V}$, $u$ is a maximal neighbor of $v$ if and only if $v$ is a maximal neighbor of $u$.
\end{lemma}

\begin{lemma}\label{npreq}
 For distinct $v, u, u' \in \mathcal{V}$, $u \preceq_v u' \Rightarrow u \npreceq_{u'} v$ .
\end{lemma}

A \emph{wheel graph} \cite{Tutte01} of order $n$ or simply an \emph{$n$-wheel}, $n \geq 3$, is a simple graph which consists of cycle of order $n$ and another vertex called the \emph{hub} such that every vertex of the cycle is connected to the hub. The vertices on the cycle are called the \emph{rim vertices}. An edge joining a rim vertex and the hub is called a \emph{spoke}, and an edge joining two consecutive rim vertices is called a \emph{rim edge}. By Theorem \ref{jack}, it follows that a wheel is globally rigid. 

The most crucial part of our algorithm is the construction of a special structure called the \emph{communication wheel}. The definition of communication wheel closely resembles to that of \emph{sensing wheel} used in \cite{Sau09}, where the authors devised a wheel based centralized sequential localization algorithm for a restricted class of sensing covered networks over a convex region.

\textbf{Communication wheel:} For any interior node $v \in \mathcal{V}$, we define a \emph{communication wheel} of $v$ as a subgraph $W$ of $\mathcal{G}$ such that 
\begin{enumerate}
 \item $W$ is a wheel graph with $v$ as the hub and the rim nodes $\{v_1, \ldots, v_m\}$ being maximal neighbors of $v$
 
%  \item $\partial(\mathcal{Z}(v)) \subset \bigcup\limits_{u \in \mathcal{N}(v)} \mathcal{Z}(u)$.
 
 \item $CCW(v, v_i) \in \mathcal{Z}(v_{i+1})$ and $CW(v, v_i) \in \mathcal{Z}(v_{i-1})$, for $i = 1, \ldots, m$, where $v_{m+1}$ means $v_1$ and $v_0$ means $v_m$.
\end{enumerate}

For a rim node $v'$ of a communication wheel $W$ of $v$, we can denote the two neighboring rim nodes of $v'$ as $CCW_{W}(v')$ and $CW_{W}(v')$ so that $CCW(v, v') \in \mathcal{Z}(CCW_{W}(v'))$ and $CW(v, v') \in \mathcal{Z}(CW_{W}(v'))$.

% Proofs of the following theorems are given in Appendix \ref{app1} and the supporting lemmas are easy to see. 

\begin{figure}[htb!]
\centering
\fontsize{10pt}{10pt}\selectfont
\def\svgwidth{0.65\textwidth}
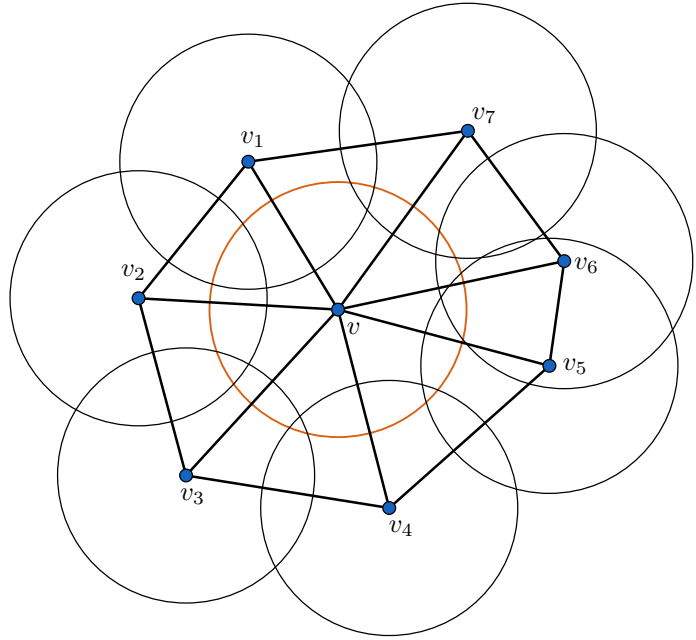
\caption{A communication wheel of $v$ with rim nodes $v_1, v_2, v_3, v_4, v_5, v_6$ and $v_7$.}
\label{communication_wheel_fig}
\end{figure}

%  
% Lemma \ref{lemma: wheel cover}-\ref{1_implies_eclipsed} are easy to see. The proofs of Theorems \ref{communication_wheel_from_non_eclipsed}-\ref{maintheorem} are in Appendix \ref{app1}. 

\begin{lemma}\label{lemma: wheel cover}
 If $W$ is a communication wheel of $v$, then $\partial(\mathcal{Z}(v)) \subset \bigcup\limits_{u \in \mathcal{V}(W) \setminus \{v\}} \mathcal{Z}(u)$.
\end{lemma}

\begin{proof}
 Follows immediately from the definition of communication wheel.
\end{proof}

\begin{theorem}\label{communication_wheel_from_non_eclipsed}
 If $v \in \mathcal{V}$ is an interior node and $v_1$ a maximal neighbor of $v$, then $v$ has a communication wheel $W$ having $v_1$ as a rim node. 
\end{theorem}

\begin{proof}
 First observe that for any maximal neighbor $v'$ of $v$, $|\partial(\mathcal{Z}(v)) \cap \partial(\mathcal{Z}(v'))| = 2$, i.e., $CW(v, v')$ and $CCW(v, v')$ are distinct points. If not, then suppose that $v'$ is a maximal neighbor of $v$ such that $\partial(\mathcal{Z}(v))$ and $\partial(\mathcal{Z}(v'))$ intersect at a single point, say $P$. Then by Lemma \ref{perimeter_2cover}, there is another neighbor of $v$, say $v''$, such that $P \in \mathcal{Z}(v'')$. Hence we have $v'' \neq v'$, such that $\partial(\mathcal{Z}(v)) \cap \mathcal{Z}(v') = \{P\} \subset \partial(\mathcal{Z}(v)) \cap \mathcal{Z}(v'')$. This contradicts the fact that $v'$ is a maximal neighbor of $v$.
  
 Now take any maximal neighbor $v_1$ of $v$. By Lemma \ref{perimeter_2cover}, choose a maximal $v_2 \in \mathcal{N}(v) \setminus \{v_1\}$, such that $CCW(v, v_1) \in \mathcal{Z}(v_2)$ and $CCW(v, v_2) \notin \mathcal{Z}(v_1)$. Notice that $CW(v, v_1) \notin \mathcal{Z}(v_2)$, because otherwise $v_1 \preceq_{v} v_2$. Since $CCW(v, v_2) \notin \mathcal{Z}(v_1)$, by again invoking Lemma \ref{perimeter_2cover}, we can choose a maximal $v_3 \in \mathcal{N}(v) \setminus \{v_1, v_2\}$, such that $CCW(v, v_2) \in \mathcal{Z}(v_3)$ and $CCW(v, v_3) \notin \mathcal{Z}(v_2)$. Continuing in this manner, after some $m$ steps we shall find $v_m \in \mathcal{N}(v) \setminus \{v_1, \ldots, v_{m-1} \}$, such that $CCW(v, v_{m-1}) \in \mathcal{Z}(v_m)$ and $CW(v, v_1) \in \mathcal{Z}(v_m)$. It is easy to see that a communication wheel of $v$ can be formed with $\{v_1, \ldots, v_m\}$ as rim nodes.
\end{proof}

\begin{corollary}\label{communication_wheel_exists}
 $v \in \mathcal{V}$ is an interior node if and only if it has a communication wheel. 
\end{corollary}

\begin{lemma}\label{r_neighbor_is_adj_to_r_neighbor}
 Let $v \in \mathcal{V}$ be an interior node and $W$ be a communication wheel of $v$. If $u \in \mathcal{V}$ is a neighbor of $v$, then $u$ is either a rim node of $W$ or adjacent to some rim node of $W$.
\end{lemma}

\begin{proof}
 Easy to see.
\end{proof}

\begin{lemma}\label{1_implies_eclipsed}
 Let $v \in \mathcal{V}$ be an interior node and $W$ be a communication wheel of $v$. If $u \in \mathcal{V}$ is a neighbor of $v$, which is adjacent to exactly one rim node of $W$, say $u'$, then $u \preceq_v u'$.
\end{lemma}

\begin{proof}
 Easy to see.
\end{proof}

% 
% 
% \begin{lemma}\label{1_implies_eclipsed}
%  For $u,v,w \in \mathcal{V}$, if  be an interior node and $W$ be a communication wheel of $v$. If $u \in \mathcal{V}$ an neighbor of $v$, which is adjacent to exactly one rim node of $W$, say $u'$, then $u \preceq_v u'$.
% \end{lemma}

%   \begin{lemma}\label{triangle}
%  Let $u, v_1, v_2, v_3 \in \mathcal{V}$ and they are all adjacent to each other. Let $\Phi(u) = U, \Phi(v_i) = V_i$, for $i = 1, 2, 3.$ If $u$ is adjacent to $v_1$ and $v_2$, and $U$ lies inside the triangle $\triangle V_1V_2V_3$, then $u$ is also adjacent to $v_3$. 
% \end{lemma}

% \begin{proof}
%  Since $U$ lies inside the triangle $\triangle V_1V_2V_3$, $U$ can be written as a convex combination of $V_1,V_2$ and $V_3$, i.e., $U =$ $\lambda_1V_1 + \lambda_2V_2 + \lambda_3V_3$, for some $\lambda_1, \lambda_2, \lambda_1V_3 \geq 0$, with $\lambda_1 + \lambda_2 + \lambda_3 = 1$. Now we have, 
%  
%  \begin{align*}
%    \|U - V_3\| &\leq \|\lambda_1V_1 + \lambda_2V_2 + \lambda_3V_3 -(\lambda_1 + \lambda_2 + \lambda_3)V_3\| \\
%   &= \|\lambda_1(V_1-V_3) + \lambda_2(V_2-V_3) + \lambda_3(V_3-V_3)\| \\
%   &\leq  \lambda_1\|V_1-V_3\| + \lambda_2\|V_2-V_3\| \\
%   &\leq \lambda_1r + \lambda_2r  \\
%   &\leq r(\lambda_1 + \lambda_2) \\
%   &\leq r
%   \end{align*}
%   
%   Hence, $u$ is adjacent to $v_3$. $\square$
% \end{proof}
%  

\begin{lemma}\label{maintheorem1}
 Let $v \in \mathcal{V}$ be a strongly interior node and $u$ a neighbor of $v$. If $W_1$ is a communication wheel of $v$, then there is a globally rigid subgraph of $\mathcal{G}$ containing $v, u$ and $W_1$.
 
\end{lemma}

\begin{proof}
  
  If $u$ is a rim node of $W_1$, then we are done, since a wheel graph is globally rigid. So, suppose that $u$ is not a rim node of $W_1$.
 
%  Let $\Phi(u) = U$. Consider the ray $\overrightarrow{VU}$. By lemma , $\overrightarrow{VU}$ intersects a wheel edge between two consecutive rim nodes, say $v_1$ and $v_2$. By the proof of lemma \ref{r_neighbor_is_adj_to_r_neighbor}, $u$ is adjacent to $v_1$ or $v_2$ or both. If $u$ is adjacent both $v_1$ and $v_2$, then $u$ can added to $W$ to form a globally rigid graph. So now assume that $u$ is adjacent to only $v_2$.
 
 Then by Lemma \ref{r_neighbor_is_adj_to_r_neighbor}, $u$ is adjacent to a rim node of $W_1$, say $v_i$. If $u$ is adjacent to another rim node, then $u$ can be added to $W_1$ to form a globally rigid graph. Hence, we assume that $u$ is adjacent to only one rim node of $W_1$, i.e., $v_i$. Then by Lemma \ref{1_implies_eclipsed}, we have $u \preceq_v v_i$.
 
 Since $v$ is a strongly interior node and $v_i$ is a neighbor of $v$, $v_i$ must be an interior node. Also, since $v_i$ is a maximal neighbor of $v$, $v$ is also a maximal neighbor of $v_i$, by Lemma \ref{maximal_dual_relation}. Hence, by Theorem \ref{communication_wheel_from_non_eclipsed}, $v_i$ has a communication wheel $W_2$ having $v$ as a rim node.

 See Fig. \ref{configurations_fig_a}. Let $v_{i-1} = CW_{W_1}(v_i)$ and $v_{i+1} = CCW_{W_1}(v_i)$. Also, let $v_{i+1}' = CW_{W_2}(v)$ and $v_{i-1}' = CCW_{W_2}(v)$. Now let $A = CCW(v, v_i) = CW(v_i, v)$ and $B = CW(v, v_i) = CCW(v_i, v)$. So we must have $A \in \mathcal{Z}(v_{i+1}) \cap \mathcal{Z}(v_{i+1}')$ and $B \in \mathcal{Z}(v_{i-1}) \cap \mathcal{Z}(v_{i-1}')$. This implies that $v_{i-1}, v_{i-1}'$ and $v_{i+1}, v_{i+1}'$ are adjacent. So $v_{i-1}$ and $v_{i+1}$ can be added to the list of rim nodes of $W_2$ and construct a wheel $W_3$ of $v_i$ having $v_{i-1}, v, v_{i+1}$ as rim nodes. Then the two globally rigid graphs $W_1$ and $W_3$ have three nodes common, namely $v_{i-1}, v, v_{i+1}$. Hence $W_1 \cup W_3$ is globally rigid. See Fig. \ref{configurations_fig_b}.

\begin{figure}[h!]
\centering
\subcaptionbox[Short Subcaption]{
        \label{configurations_fig_a}
}
[
    0.56\textwidth % width of caption
]
{
    \fontsize{10pt}{10pt}\selectfont
    \def\svgwidth{0.56\textwidth}
    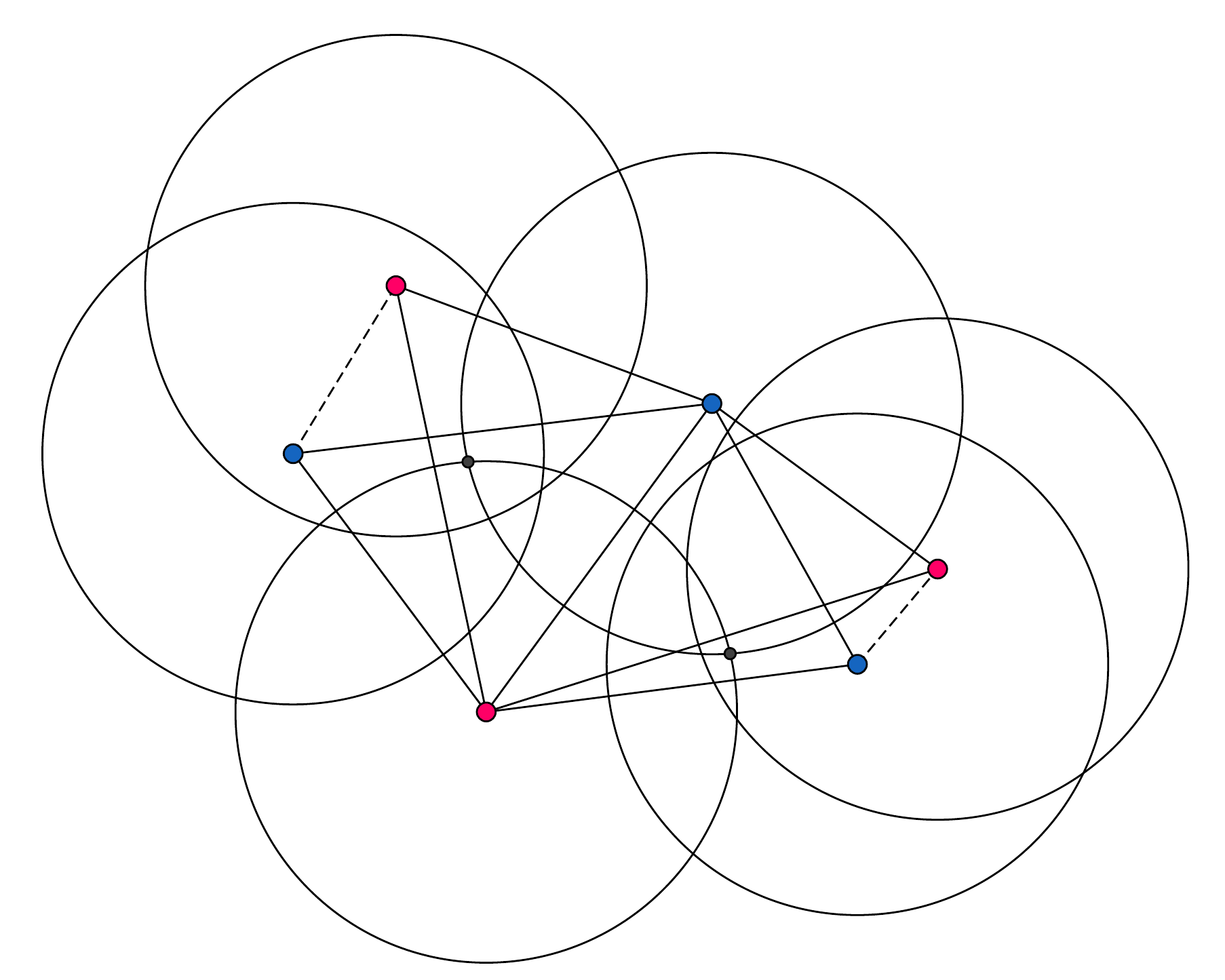
}
\hfill % seperation
\subcaptionbox[Short Subcaption]{
     \label{configurations_fig_b}
}
[
    0.43\textwidth % width of caption
]
{
    \fontsize{10pt}{10pt}\selectfont
    \def\svgwidth{0.43\textwidth}
    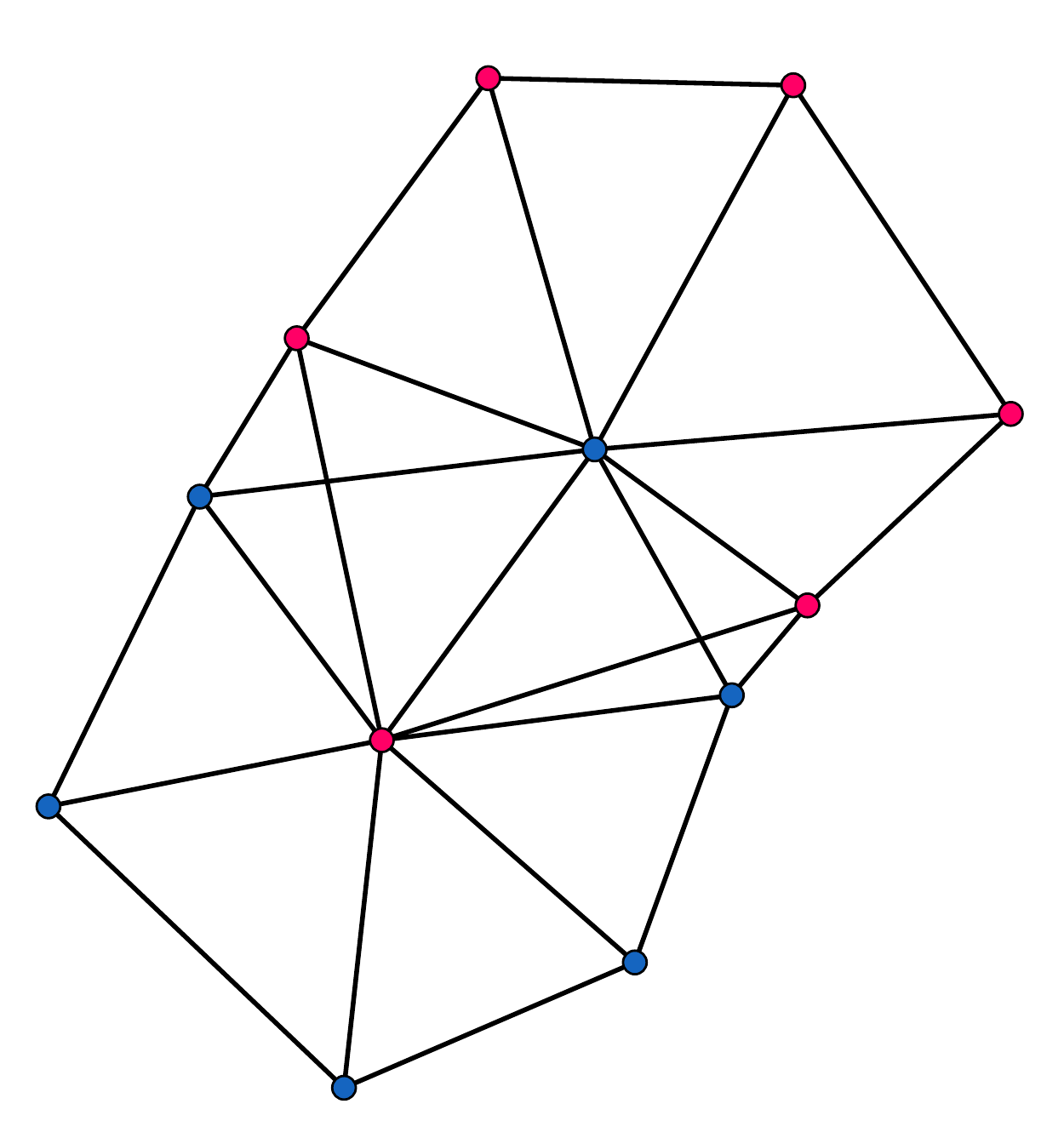
}
\caption[Short Caption]{Illustrations supporting the proof of Lemma \ref{maintheorem1}.}
\label{configurations}
\end{figure}

  Now it is sufficient to prove that $u$ is adjacent to at least three nodes of $W_1 \cup W_3$. Let $\mathcal{R}im(W_3)$ $= \mathcal{V}(W_3) \setminus \{v_i\}$ be the set of rim nodes of $W_3$. We have $\partial(\mathcal{Z}(v_i)) \subset \bigcup\limits_{w \in \mathcal{R}im(W_3)} \mathcal{Z}(w)$. Since $u \preceq_v v_i$, we have $u \npreceq_{v_i} v$, by Lemma \ref{npreq}. In other words, $\partial(\mathcal{Z}(v_i)) \cap \mathcal{Z}(u) \nsubseteq \partial(\mathcal{Z}(v_i)) \cap \mathcal{Z}(v)$. Therefore, we have $\partial(\mathcal{Z}(v_i)) \cap \mathcal{Z}(u) \subset \partial(\mathcal{Z}(v_i)) \subset \bigcup\limits_{w \in \mathcal{R}im(W_3)} \mathcal{Z}(w)$ $\Rightarrow (\partial(\mathcal{Z}(v_i)) \cap Z(u)) \bigcap (\bigcup\limits_{w \in \mathcal{R}im(W_3) \setminus \{v\}} \mathcal{Z}(w))$ $\neq \emptyset$. This implies that $u$ is adjacent to some rim node of $W_3$ other than $v$. We already have assumed that $u$ is adjacent to $v$ and $v_i$. Hence, $u$ is adjacent to at least three nodes of $W_1 \cup W_3$.
  \end{proof}

\begin{theorem}\label{maintheorem2}
 If $v \in \mathcal{V}$ is a strongly interior node, then there is a subgraph $\mathcal{H}_v$ of $\mathcal{G}$ containing $v$ such that 1) $\mathcal{H}_v$ contains all neighbors of $v$, 2) $\mathcal{H}_v$ is globally rigid.
\end{theorem}

\begin{proof}

 Let $W_1$ be a communication wheel of $v$. Then by Lemma \ref{maintheorem1}, for each neighbor of $v$ not in $W_1$, we obtain a globally rigid subgraph of $\mathcal{G}$ containing the neighbor, $v$ and $W_1$. So any two of these globally rigid subgraphs have at least three nodes in common. Hence, these graphs constitute to form the desired globally rigid graph. 
\end{proof}

\begin{theorem}\label{maintheorem}
 If the strong interior of $\mathcal{G}$ is connected, then $\mathcal{G}$ has a globally rigid subgraph $\mathcal{R}$ which contains 1) all strongly interior nodes, 2) all non-isolated weakly interior nodes. 
\end{theorem}
 
\begin{proof}
 
  Choose any strongly interior node $v$. Denote the subgraph of $\mathcal{G}$ consisting of only the node $v$ as $\mathcal{R}_0$. Since the strong interior of $\mathcal{G}$ is connected, every strongly interior node of $\mathcal{G}$ is connected to $v$ by a path consisting strongly interior nodes. The distance of a strongly interior node from $v$ is defined as the smallest length of such a path. Let $m$ be the maximum distance of a strongly interior node from $v$. We shall prove the theorem by inductively constructing globally rigid subgraphs $\mathcal{R}_0, \mathcal{R}_1, \ldots, \mathcal{R}_m$, where $\mathcal{R}_j$ contains all strongly interior nodes at a distance at most $j$ from $v$. $\mathcal{R}_0$ is globally rigid as it is only a singleton node. $\mathcal{R}_1$ is constructed using Theorem \ref{maintheorem2}. Suppose that $\mathcal{R}_0, \mathcal{R}_1, \ldots, \mathcal{R}_j$, $1 \leq j < m$, are already constructed. Now consider a strongly interior node $v'$ at a distance $j+1$ from $v$. In a smallest path from $v$ to $v'$, let $v'$ be adjacent to $v''$. Clearly $v''$ is in $\mathcal{R}_j$. Since $\mathcal{R}_j$ is globally rigid, $v''$ is adjacent to at least three nodes $v_1, v_2, v_3$ in $\mathcal{R}_j$ ($\mathcal{R}_j$ has at least four nodes as it contains the communication wheel of $v$). By Theorem \ref{maintheorem2}, there is a globally rigid graph containing $v''$ and its neighbors $v', v_1, v_2, v_3$. Union of this graph and $\mathcal{R}_j$ is globally rigid as there are at least three nodes in common, namely $v'', v_1, v_2, v_3$, etc. Similarly for each strongly interior node at a distance $j+1$ from $v$, we extend the subgraph $\mathcal{R}_j$ preserving global rigidity to eventually obtain a globally rigid graph $\mathcal{R}_{j+1}$ containing all the strongly interior nodes at a distance at most $j+1$ from $v$. The inductive argument leads to the globally rigid subgraph $\mathcal{R}_m$ which contains all strongly interior nodes in $\mathcal{G}$. Each non-isolated weakly interior node $v'''$ is adjacent to some strongly interior node in $\mathcal{R}_m$. For each such $v'''$, again by the same construction, we can extend $\mathcal{R}_m$ preserving global rigidity to include $v'''$, if it is not already in $\mathcal{R}_m$. The resulting graph is the desired globally rigid subgraph $\mathcal{R}$. 
\end{proof}

\section{The Localization Algorithm}\label{sec_algo}

In the beginning, each node messages its neighbor list along their distances form itself to all its neighbors. Therefore, every node $u \in \mathcal{V}$ knows the neighbors of all its neighbors and also if $v$ is a neighbor of $u$ and $w$ is a neighbor of $v$, then $u$ knows $d(v,w)$ as well.
The three main stages of our algorithm are 1) construction of communication wheel, 2) leader election, and 3) propagation. They are discussed in detail in the following subsections.

\begin{figure}[h!]
\centering
\subcaptionbox[Short Subcaption]{
        \label{configurations_fig_a}
}
[
    0.45\textwidth % width of caption
]
{
    \fontsize{10pt}{10pt}\selectfont
    \def\svgwidth{0.46\textwidth}
    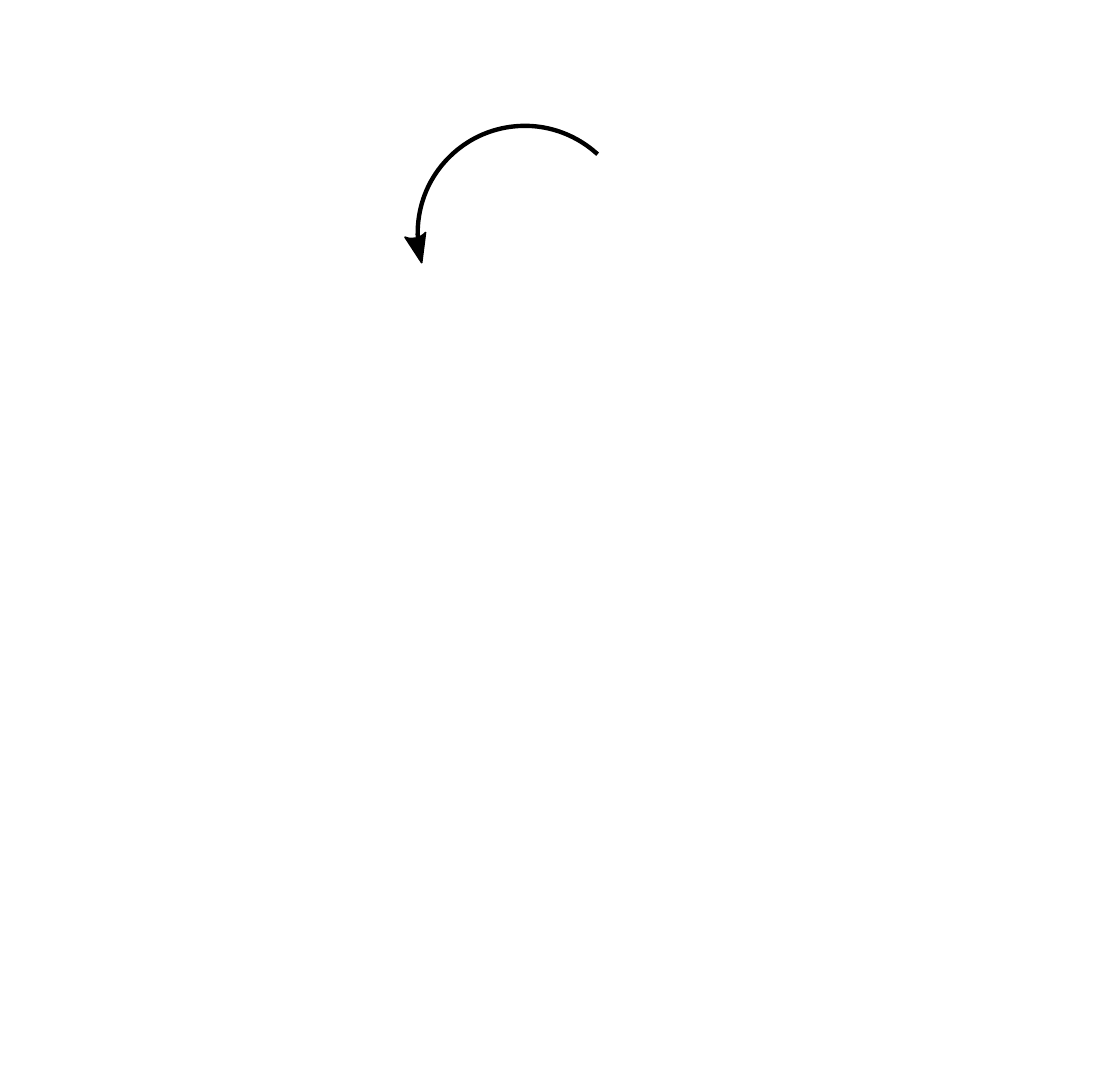
}
\hfill % seperation
\subcaptionbox[Short Subcaption]{
     \label{configurations_fig_b}
}
[
    0.4\textwidth % width of caption
]
{
    \fontsize{10pt}{10pt}\selectfont
    \def\svgwidth{0.4\textwidth}
    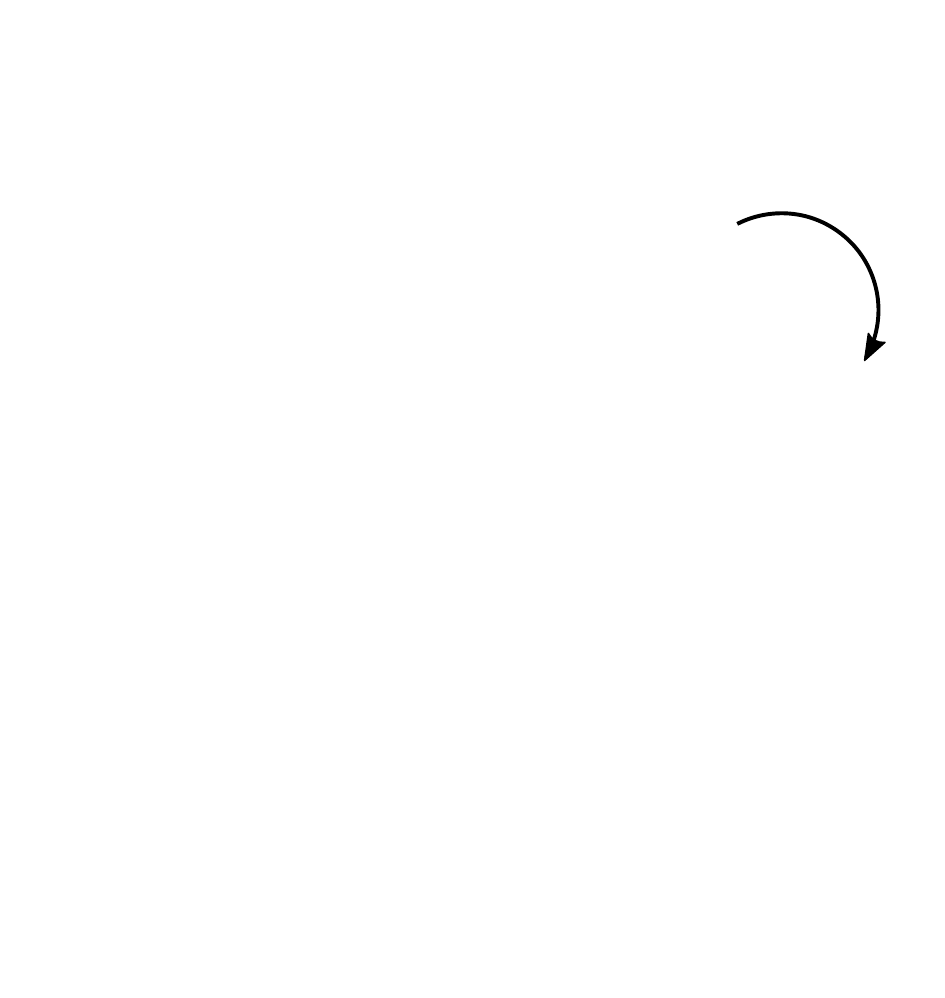
}
\caption[Short Caption]{A node $v$ executing Algorithm \ref{wheel_algo} sets its local coordinate system in such a way that $w_1$ gets positive $Y$-coordinate.}
\label{configurations}
\end{figure}

\begin{algorithm}[H]
    \setstretch{0.1}
%     \SetKwInOut{Input}{Input}
%     \SetKwInOut{Output}{Output}
    \SetKwProg{Fn}{Procedure}{}{}
    \SetKwRepeat{Do}{do}{while}

% %     \underline{function Euclid} $(a,b)$\;
%     \Input{Two nonnegative integers $a$ and $b$ A two-dimensional digital image is a grid of pixels, each assigned with gray values or color values. A two-dimensional digital image is giv}
%     \Output{$\gcd(a,b)$}
    
    \nonl \emph{\{The node $v$ constructs a communication wheel. If it successfully constructs a communication wheel, it declares itself as an interior node, or otherwise a boundary node.\}}
    
    \nonl \Fn{\textsc{CommunicationWheel}($v$)}{
    
    $w = [\ ]$  \; 
    
    $v.position$ $\longleftarrow$ origin\;

    $w_{0}$ $\longleftarrow$ closest neighbor of $v$ \;
    
    $w_{0}.position$ $\longleftarrow$ on the $X$-axis according to the distance between $v$ and $w_{0}$\;

    $w_{1}$ $\longleftarrow$ the common neighbor of $v$ and $w_{0}$ closest to $v$ that covers a boundary-intersection of $v$ and $w_{0}$, such that $w_{1} \npreceq_{v} w_{0}$ \;
            
    \eIf{(no such $w_1$ is found)}
      {
        $v.type$ $\longleftarrow$ \emph{boundary}\;
        \textbf{break}\;
      }
      {
        $w_{1}.position$ $\longleftarrow$ one of the two possible positions preserving the distances from $v$ and $w_{0}$ such that $w_{1}$ has positive $Y$-coordinate\;
        
        $CCW(v, w_{0})$ $\longleftarrow$ the boundary-intersection of $v$ and $w_{0}$ covered $w_{1}$
        
        $i = 1$ \;
        
        \Do{($w_{i} \neq$ null $\&$ $w_{i}$ does not cover $CW(v,w_0$))}{
        
        $i++$\;
        
        \textsc{NextRim}($v$,$w_{i-1}$,$w_{i-2}$)\;
        
        }
        
      \eIf{($w_{i} =$ null)}
      {
	$v.type$ $\longleftarrow$ \emph{boundary}\;
      }
      {
	$v.type$ $\longleftarrow$ \emph{interior}\;
      }
      
      }
      }
    \caption{\textsc{CommunicationWheel}}
    \label{wheel_algo}
\end{algorithm}
% 
% \end{minipage}
% \end{figure}

  \begin{algorithm}[H]
    \setstretch{0.1}
%     \SetKwInOut{Input}{Input}
%     \SetKwInOut{Output}{Output}
    \SetKwProg{Fn}{Function}{}{}
    \SetKwRepeat{Do}{do}{while}

% %     \underline{function Euclid} $(a,b)$\;
%     \Input{Two nonnegative integers $a$ and $b$ A two-dimensional digital image is a grid of pixels, each assigned with gray values or color values. A two-dimensional digital image is giv}
%     \Output{$\gcd(a,b)$}
    
    \nonl \emph{\{Given two consecutive rim nodes $w_{i-1}$ and $w_{i-2}$, the node $v$ searches for the next rim node.\}}
    
    \nonl \Fn{\textsc{NextRim}($v$,$w_{i-1}$,$w_{i-2}$)}{
    
    $D$ $\longleftarrow$ $r$ \;
    
    \For{($u \in \mathcal{N}(w_{i-1}) \cap \mathcal{N}(v)$)}{
    
    \If{(distance between $u$ and $v$ $\leq D$)}{
    
    \eIf{($u$ is adjacent to $w_{i-2}$)}{
    
      Find the unique position of $u$ using its distances from $v, w_{i-1}$ and $w_{i-2}$\; 
      
      \If{($u$ covers $CCW(v,w_{i-1})$ and $CCW(v, u)$ is not covered by $w_{i-1}$)}{
      
      $w_{i} \longleftarrow u$\;
      
      $w_{i}.position \longleftarrow$ the unique position of $u$ determined using its distances from $v, w_{i-1}$ and $w_{i-2}$\; 
      
      $D \longleftarrow$ the distance between $u$ and $v$\;
      }
    }
    {
      Find two possible positions of $u$ using its distances from $v$ and $w_{i-1}$\;
      
      \If{(for any of the two possible positions, $u$ covers $CCW(v,w_{i-1})$ and $CCW(v, u)$ is not covered by $w_{i-1}$)}{
      
      $w_{i} \longleftarrow u$\;

      $w_{i}.position \longleftarrow$ the one of the two possible positions of $u$ for which it covers $CCW(v, w_{i-1})$\; 
      
      $D \longleftarrow$ the distance between $u$ and $v$\;

      }
    }

    }
    }

    }
    \caption{\textsc{NextRim}}
    \label{next_rim_algo}
\end{algorithm}

\subsection{Construction of Communication Wheel}\label{sec: wheel cons}

Each sensor node $v$ starts off computations by executing the \textsc{CommunicationWheel} algorithm. The algorithm finds if the node is interior or boundary, and also constructs a communication wheel if it is interior. A pseudocode description of the procedure is presented in Algorithm \ref{wheel_algo}. The algorithm \textsc{CommunicationWheel} is similar to the constructions used in the proof of the Theorem \ref{communication_wheel_from_non_eclipsed}. To construct a communication wheel of a node $v$, if it exists, we first need to find a maximal neighbor. In view of Lemma \ref{eclipse_distance_relation}, the closest neighbor of a node is guaranteed to be a maximal neighbor. After finding the closest neighbor, call it $w_0$, $v$ assigns its position on the $X$-axis and itself at the origin. Then it searches for a common neighbor of $v$ and $w_0$ that covers a boundary-intersection of $v$ and $w_{0}$. If no such node is found, then $v$ is a boundary node. If more than one of such nodes are found, the one closest to $v$ is to be taken. Let us call this node $w_1$. Now the distance of $w_1$ from $v$ and $w_0$ is known. From this data, there are two possible coordinates for $w_1$, one with positive $Y$-coordinate and one with negative $Y$-coordinate. Choose the position for $w_1$ so that its $Y$-coordinate is positive. In other words, $v$ sets its local coordinate system in such a way  that $w_1$ gets positive $Y$-coordinate. Also set the boundary-intersection of $v$ and $w_{0}$ that is covered by $w_1$ as $CCW(v, w_0)$. In other words $v$ sets `counterclockwise' to be the direction in which if one rotates a ray, from the origin towards the positive direction of the $X$-axis, by $\frac{\pi}{2}$, it coincides with the positive direction of the $Y$-axis. While discussing Algorithm \ref{wheel_algo}, `counterclockwise' and `clockwise' will always be with respect to the local coordinate system of node executing the algorithm. Note that since $w_0$ is a maximal neighbor of $v$, $CW(v, w_0)$ is not covered by $w_1$. After fixing the positions of $w_0$ and $w_1$, the subroutine \textsc{NextRim} is recursively called to find the subsequent rim nodes of the communication wheel. Given two consecutive rim nodes $w_{i-1}$ and $w_{i-2}$, having positions fixed, \textsc{NextRim}($v$,$w_{i-1}$,$w_{i-2}$) finds the next rim node $w_i$. The program terminates when either \textsc{NextRim} reports a failure or returns a node that covers $CW(v, w_0)$. A pseudocode description of the \textsc{NextRim} function is presented in Algorithm \ref{next_rim_algo}.

% The correctness of the algorithm is proved in Appendix \ref{correctness}. It is straight-forward to verify that the worst-case time complexity of the algorithm is $O(\Delta^3)$, where $\Delta$ is the maximum degree of a node in the network. 
% 
% 

% \section{Correctness of Algorithm \ref{wheel_algo}}\label{correctness}

\begin{theorem}\label{communication_wheel_correctness}
 The algorithm \textsc{CommunicationWheel} is correct, i.e., if $v$ is an interior node then \textsc{CommunicationWheel}$(v)$ constructs a communication wheel of $v$ and declares it as an interior node; and otherwise declares it as a boundary node.  
\end{theorem}

\begin{proof}
 The algorithm replicates the proof of Theorem \ref{communication_wheel_from_non_eclipsed}. The algorithm starts off with fixing a maximal neighbor of $v$ as the first rim node $w_0$. Then it recursively finds the rim nodes $w_i$ such that $w_i$ covers $CCW(v,w_{i-1})$ and $CCW(v, w_i)$ is not covered by $w_{i-1}$. The algorithm terminates when there is no such $w_i$ or when $w_{i}$ covers $CW(v,w_0)$. Thus in view of the proof of Theorem \ref{communication_wheel_from_non_eclipsed}, we only need to show that these steps are correctly executed.

 \paragraph{Initialization:} The closest neighbor of $v$ is set as $w_0$. Hence, $w_0$ is a maximal neighbor of $v$ by Lemma \ref{eclipse_distance_relation}. Note that in order to compute the communication wheel of $v$, if it exists, first we need to fix the positions of (i.e., assign virtual coordinates to) at least three nodes of the communication wheel, preserving their mutual distances. So, first $v$ is assigned with virtual coordinates $(0,0)$. If the distance between $v$ and $w_0$ is $d$, then the coordinates of $w_0$ are set as $(d, 0)$. Now we have to check if there is a common neighbor $w$ of $v$ and $w_0$ that covers a boundary intersection of $v$ and $w_0$, and such that $w \npreceq_{v} w_{0}$. For any common neighbor $w$ of $v$ and $w_0$, this can be easily checked from $d(v, w_0), d(v, w)$ and $d(w, w_0)$. If no such $w$ is found, then $v$ is obviously a boundary node. Otherwise one such node that is closest to $v$ is set as the next rim node $w_1$. From $d(v, w_1)$ and $d(w_0, w_1)$, two possible coordinates of $w_1$ can be found, one with positive $Y$-coordinate and one with negative $Y$-coordinate. Then $v$ sets its local coordinate system in such a way that $w_1$ gets positive $Y$-coordinate, and hence $w_1$ covers $CCW(v, w_0)$. Recall that here counterclockwise and clockwise is defined with respect to the local coordinate system of $v$ as described in Section \ref{sec: wheel cons}.
 
 \paragraph{Recursion:} After $w_0$ and $w_1$ are fixed, the algorithm will recursively call \textsc{NextRim}($v$,$w_{i-1}$,$w_{i-2}$) to find the next rim node $w_{i}$, if it exists. In the for loop (line 2 in Algorithm \ref{next_rim_algo}), the common neighbors of $v$ and $w_{i-1}$ are scanned through to find nodes $u \in \mathcal{N}(w_{i-1}) \cap \mathcal{N}(v)$ such that $u$ covers $CCW(v,w_{i-1})$ and $CCW(v, u)$ is not covered by $w_{i-1}$. Among these nodes, the one closest to $v$ is set as the next rim node $w_i$. If no such node is found, then $v$ is a boundary node. Notice that in order to check whether $u$ covers $CCW(v,w_{i-1})$ or not, the exact position of $u$ needs to be known. As we scan through $\mathcal{N}(w_{i-1}) \cap \mathcal{N}(v)$, there are two cases to consider:

   \textbf{Case 1.} Suppose that $u$ is adjacent to $w_{i-2}$. Now the positions of $v, w_{i-1}$ and $w_{i-2}$ are known. Hence the position of $u$ can ascertained from its distances from $v, w_{i-1}$ and $w_{i-2}$. Once the position of $u$ is found, it can be checked whether it covers $CCW(v,w_{i-1})$ and also whether $w_{i-1}$ covers $CCW(v, u)$.  
  
  \textbf{Case 2.} Suppose that $u$ is not adjacent to $w_{i-2}$. Two possible positions of $u$ can be found from its distance from $v$ and $w_{i-1}$. Call these two possible positions $U_1$ and $U_2$. $U_1$ and $U_2$ are mirror images of each other with respect to the line joining $v$ and $w_{i-1}$. $CCW(v,w_{i-1})$ and $CW(v,w_{i-1})$ are also mirror images of each other with respect to the line joining $v$ and $w_{i-1}$. Hence if a node at $U_1$ covers $CCW(v,w_{i-1})$, then a node at $U_2$ covers $CW(v,w_{i-1})$ as well. Similarly if a node at $U_1$ covers neither $CCW(v,w_{i-1})$ nor $CW(v,w_{i-1})$, then the same is true for a node at $U_2$. So consider the following three possibilities:  
  
   \textbf{Case 2a.} If for both positions $U_1$ and $U_2$, no boundary intersection between $v$ and $w_{i-1}$ is covered, then $u$ does not meet the desired criteria that it is to cover $CCW(v,w_{i-1})$.
   
   \textbf{Case 2b.} Suppose that for both positions $U_1$ and $U_2$, both of the boundary intersections between $v$ and $w_{i-1}$ are covered. This can not happen as this would imply that $w_{i-i} \preceq_v u$ and hence is adjacent to $w_{i-2}$, contradicting our assumption.
   
   \textbf{Case 2c.} Suppose that $u$, if situated at $U_1$, covers $CCW(v,w_{i-1})$ (and not $CW(v,w_{i-1})$). Hence $u$, if it is at $U_2$, would cover $CW(v,w_{i-1})$ and would not cover $CCW(v,w_{i-1})$. In this case, the algorithm determines the position of $u$ to be $U_1$. We shall prove that $U_1$ is indeed the correct position of $u$. Suppose on the contrary that the actual position of $u$ is $U_2$. First observe that $CW(v,w_{i-1}) \in \mathcal{Z}(w_{i-2}) \cap \partial(\mathcal{Z}(v))$. Otherwise, it implies that $w_{i-2} \preceq_v w_{i-1}$. We argue that this is impossible. For $i = 2$, it is obvious since $w_0 \npreceq_v w_1$. Recall that $w_0$ is the closest neighbor, and hence is a maximal neighbor, of $v$. For $i > 2$, we assume as induction hypothesis that rim nodes $w_{i-1}, \ldots, w_1, w_0$ are successfully found by the algorithm. Also each $w_{j}$, for $j = 1, \ldots, i-1$,  must satisfy the two criteria: 1) $w_j$ covers $CCW(v, w_{j-1})$ and 2) $CCW(v, w_j)$ is not covered by $w_{j-1}$. In fact, as mentioned earlier, the algorithm chooses as $w_j$ the closest among the nodes that satisfy these two criteria. So in particular, $w_{i-2}$ is the closest neighbor of $v$ such that 1) $w_{i-2}$ covers $CCW(v, w_{i-3})$ and 2) $CCW(v, w_{i-2})$ is not covered by $w_{i-3}$. Clearly if $w_{i-2} \preceq_v w_{i-1}$, $w_{i-1}$ also satisfies the two aforesaid conditions. But by Lemma \ref{eclipse_distance_relation}, $w_{i-1}$ is closer to $v$ than $w_{i-2}$. This contradicts the fact that $w_{i-2}$ is the closest node satisfying the two aforesaid conditions. So we have $CW(v,w_{i-1}) \in \mathcal{Z}(w_{i-2})$. Also $CW(v,w_{i-1}) \in \mathcal{Z}(u)$ as $u$ is assumed to be at $U_2$. Hence $\mathcal{Z}(w_{i-2}) \cap \mathcal{Z}(u) \neq \emptyset$. This is a contradiction as $u$ and $w_{i-2}$ are not adjacent.

 \paragraph{Termination:} It is easy to see that the algorithm terminates. 
\end{proof}

\subsection{Leader Election}

Once a node identifies itself as interior or boundary, it announces the result to all its neighbors. Hence, every node can determine if it is a strongly interior node or not. Since the strong interior is connected and the nodes have unique id's, the strongly interior nodes can elect a leader among themselves by executing a leader election protocol \cite{lynch}. 

\subsection{Propagation}

Starting from the leader, different nodes will gradually get localized via message passing. The correctness of the process will follow from the discussions in this subsection and the proofs of Theorem \ref{maintheorem2} and \ref{maintheorem}. There are five types of messages that a sensor node can send to another node:

\begin{enumerate}
 \item ``$I\ am\ at\ \ldots$''
 \item ``$You\ are\ at\ \ldots$''
 \item ``$Construct\ wheel\ with\ me\ at\ \ldots\ and\ v\ at\ \ldots$''
 \item ``$Construct\ wheel\ with\ me\ at\ \ldots,\ you\ at\ \ldots,\ v\ at\ \ldots\ and\ find\ u$''
 \item  ``$u\ is\ at\ \ldots$''.  
%  \item ``''
%  \item ``''
% ``$u\ is\ at\ a\ distance\ d\ from\ (x,y)$''  
\end{enumerate}

% The localization starts from the leader $v_l$ as it initiates the localization of $\mathcal{H}_{v_l}$ (See Theorem \ref{maintheorem2}).

The nodes of the network will be localized in the local coordinate system of the leader $v_l$ set during its execution of Algorithm \ref{wheel_algo}. Henceforth, this coordinate system will be referred to as the \emph{global coordinate system}. So the leader first localizes itself by setting its coordinates to $(0,0)$. Any non-leader node $u$ is localized by either receiving a  ``$You\ are\ at\ \ldots$'' message or receiving at least three ``$I\ am\ at\ \ldots$'' messages. In the first case, some node has calculated the coordinates of $u$ and has sent it to $u$. In the second case, $u$ receives the coordinates of at least three neighbors and therefore, can calculate its own coordinates. When a node is localized, it announces its coordinates to all its neighbors. After setting its coordinates to $(0,0)$ and $v_l$ initiates the localization of $\mathcal{H}_{v_l}$ (See Theorem \ref{maintheorem2}). It first announces its coordinates to all its neighbors via the message ``$I\ am\ at\ (0,0)$''. During the construction of its communication wheel, $v_l$ had assigned coordinates to the rim nodes. So $v_l$ sends these coordinates to the corresponding rim nodes via the message ``$You\ are\ at\ \ldots$''. Let us denote the communication wheel of $v_l$ as $\mathcal{W}(v_l)$ and the set of all rim nodes as $\mathcal{R}im(v_l)$. When a rim node receives this message, it sets its coordinates accordingly and announces it to all its neighbors via the message ``$I\ am\ at\ \ldots$''. Notice that a rim node does not need to send this message to $v_l$. There are multiple such modifications that can be made to reduce the number of messages used in the algorithm. But we do not mention them for simplicity of the description. Now if a neighbor of $v_l$ is adjacent to at least two nodes of $\mathcal{R}im(v_l)$, then it can localize itself, since it will receive ``$I\ am\ at\ \ldots$'' messages from at least three nodes, i.e., one from $v_l$ and at least two from $\mathcal{R}im(v_l)$. But if a neighbor of $v_l$ is adjacent to only one vertex from $\mathcal{R}im(v_l)$, then it may not be localized. To resolve this, $v_l$ computes $|\mathcal{N}(u) \cap \mathcal{R}im(v_l)|$ for all $u \in \mathcal{N}(v_l)$. If it finds a $u \in \mathcal{N}(v_l)$ with $\mathcal{N}(u) \cap \mathcal{R}im(v_l) = \{v_{i}\}$, it sends the message ``$Construct\ wheel\ with\ me\ at\ \ldots\ and\ v_{i+1}\ at\ \ldots$'' to $v_i$, where $v_{i+1}$ is a neighboring rim node of $v_i$ in $\mathcal{W}(v_l)$. When $v_i$ receives this message from $v_l$, it does the following. Since $v_l$ is a strongly interior node, $v_i$ must be an interior node. Therefore, $v_i$ has already computed the communication wheel $\mathcal{W}(v_i)$ and coordinates of each of its nodes with respect to its local coordinate system. Since $v_l$ is a maximal neighbor of $v_i$ (by Lemma \ref{maximal_dual_relation}), it is adjacent to at least two nodes of $\mathcal{R}im(v_i)$ (by Lemma \ref{1_implies_eclipsed}). Hence, $v_i$ can compute the coordinates of $v_l$ with respect to its local coordinate system. Let $\mathcal{W}'$ be the globally rigid graph $v_l \cup \mathcal{W}(v_i)$. Now, from the proof of Theorem \ref{maintheorem2}, it is known that $v_{i+1}$ is adjacent to at least three nodes of $\mathcal{W}'$. Hence,  $v_i$ can also compute the coordinates of $v_{i+1}$ with respect to its local coordinate system. So, $v_i$ has the coordinates of all nodes of $\mathcal{W}''$ $= v_{i+1} \cup \mathcal{W}'$ with respect to its local coordinate system. Now, $v_i$ will compute the positions of all nodes of $\mathcal{W}''$  with respect to the global coordinate system set by $v_l$. Let us call them the \emph{true positions} of the nodes. Note that $v_{i}$ knows the true positions of at least three nodes of $\mathcal{W}''$, namely, itself, $v_l$, and $v_{i+1}$. With this information, $v_{i}$ can determine the formula that transforms its local coordinate system to the global coordinate system. Hence, $v_{i}$ computes the true positions of all nodes in $\mathcal{W}''$ and informs them via  ``$You\ are\ 
at\ \dots$'' messages. Hence, all nodes in $\mathcal{W}''$ will be localized and will announce their locations to all their neighbors. Since $u$ is adjacent to at least three nodes in $\mathcal{W}''$ (from the proof of Theorem \ref{maintheorem2}), it will also get localized. Therefore, we see that every neighbor of $v_l$ eventually gets localized.

The localization propagates as each strongly interior node localizes its neighbors. However, a strongly interior node $v$ can compute the positions of its neighbors only with respect to its local coordinate system. Hence, in order to compute the true positions (i.e., to perform coordinate transformation), it needs to know its true position and that of at least two neighbors. Hence, when a localized strongly interior node $v$ receives at least two ``$I\ am\ at\ \dots$'' messages, it starts to localize its neighbors in the following way. Let $u$ be a neighbor of $v$. If $u$ is adjacent  at least two nodes of $\mathcal{R}im(v)$, then $v$ can compute the position of $u$ in terms of its local coordinate system. Otherwise, if $\mathcal{N}(u) \cap \mathcal{R}im(v) = \{v_{i}\}$, then $v$ sends the message ``$Construct\ wheel\ with\ me\ at\ \ldots,\ you\ at\ \ldots,\ v_{i+1}\ at\ \ldots\ and\ find$ $u$'' to $v_i$, where $v_{i+1}$ is a neighboring rim node to $v_i$ in $\mathcal{W}(v)$, and positions mentioned in the message are given in local coordinates of $v$. Again, as   $v_i$ is a maximal neighbor of $v$, by Lemma \ref{maximal_dual_relation} and \ref{1_implies_eclipsed}, $v$ is either in $\mathcal{W}(v_i)$ or adjacent to at least three nodes in $\mathcal{W}(v_i)$. Also, $v_{i+1}$ is either in $\mathcal{W}' = \mathcal{W}(v_i) \cup v$ or adjacent to at least three nodes in $\mathcal{W}'$ (from the proof of Theorem \ref{maintheorem2}). Hence, from the data received from $v$, $v_i$ can compute the positions of all the nodes in $\mathcal{W}'' = \mathcal{W}' \cup v_{i+1}$ in terms of the local coordinates of $v$. From the proof of Theorem \ref{maintheorem2}, $u$ is either in $\mathcal{W}''$ or adjacent to at least three nodes in $\mathcal{W}''$. So $v_i$ can compute the position of $u$ in terms of the local coordinate system of $v$, and then sends the information back to $v$ via the message ``$u\ is\ at\ \ldots$''. Hence, $v$ computes the positions of all its neighbors in its local coordinate system. So, $v$ now knows the positions of three nodes (namely, itself and the two nodes from which it has received ``$I\ am\ at\ \dots$'' message) with respect to both its local coordinate system and the global coordinate system set by the leader. Hence, $v$ can find the formula that transformations its local coordinate system to the global coordinate system. Hence, $v$ computes the true positions of all its neighbors and then informs them via ``$You\ are\ at\ \dots$'' messages. However, it still remains to prove that every non-leader localized strongly interior node $v$ always receives at least two ``$I\ am\ at\ \dots$'' messages, that triggers the propagation.  Since $v$ is localized, either it has received three ``$I\ am\ at\ \dots$'' messages or one ``$You\ are\ at\ \dots$'' message. If it is the first case, then we are done. In the later case, $v$ receives the ``$You\ are\ at\ \dots$'' message from a localized interior node, say $v'$. Then by Lemma \ref{r_neighbor_is_adj_to_r_neighbor}, $v$ is either in the communication wheel of $v'$, or adjacent to at least one of its rim nodes. Observe that $v'$ is localized and has also localized all its rim nodes. So, $v$ will get least two ``$I\ am\ at\ \dots$'' messages, as all localized nodes announce their positions. Therefore, all strongly interior and non-isolated weakly interior nodes get localized, while some some boundary nodes and isolated weakly interior nodes may not get localized. Therefore, all strongly interior and non-isolated weakly interior nodes get localized. Also, if a localized weakly interior node receives ``$I\ am\ at\ \dots$'' messages from at least two nodes, it can localize all the rim nodes and also neighbors that are adjacent to at least two rim nodes.

% 
% Then by Lemma \ref{r_neighbor_is_adj_to_r_neighbor}, $v$ is either in the communication wheel of $v'$, or adjacent to at least one of its rim nodes. Since, $v'$ is localized and has also localized all its rim nodes, $v$ will get least two ``$I\ am\ at\ \dots$'' messages, as all localized nodes announce their positions. Therefore, all strongly interior and non-isolated weakly interior nodes get localized, while some some boundary nodes and isolated weakly interior nodes may not get localized. 

% So, some boundary nodes and isolated weakly interior nodes may not get localized.

% Also, if a localized weakly interior node receives ``$I\ am\ at\ \dots$'' messages from at least two nodes, it will localize all the rim nodes and also neighbors that are adjacent to at least two rim nodes. Therefore, all strongly interior and non-isolated weakly interior nodes get localized, while some some boundary nodes and isolated weakly interior nodes may not get localized. 

\begin{figure}
\centering
\includegraphics[width=.6\textwidth]{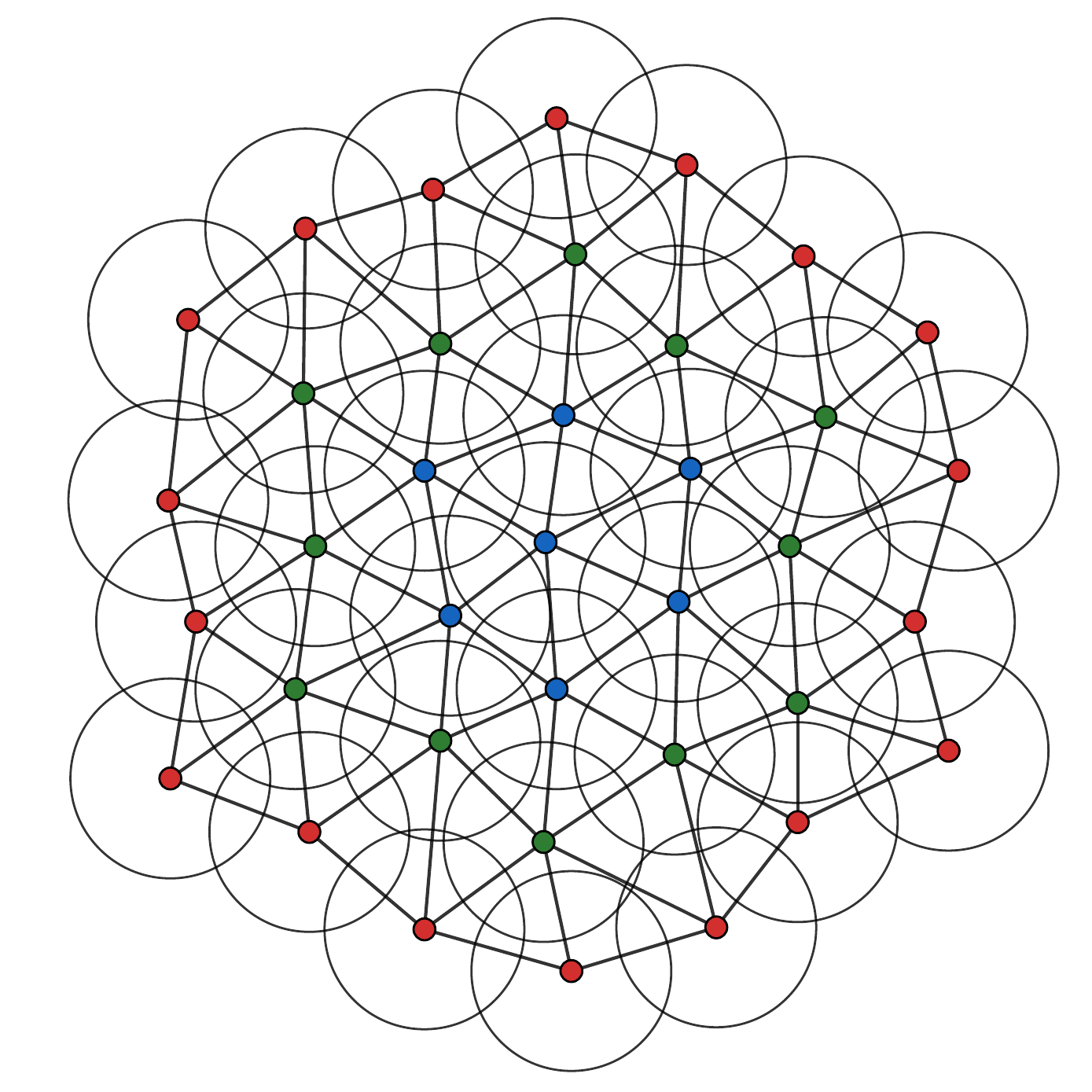}
\caption{The red, green and blue nodes are respectively boundary, weakly interior and strongly interior nodes. It is easy to see that trilateration does not progress beyond the base step for any choice of the initial triangle. However, our algorithm always localizes all the nodes of the network from any initial strongly interior node. This structure can be extended arbitrarily. Hence for any $n \in \mathbb{N}$, we have a network of size $\geq n$, such that 1) it is always entirely localized by our algorithm, 2) but trilateration fails to localize more than 3 nodes for any choice of the initial triangle.}
\label{honey}
\end{figure}

\section{Concluding Remarks}

%   The performance of trilateration substantially differs for different choices of the initial three anchors. Our algorithm starts from a strongly interior node and provided that the strong interior of the network is connected, it is guaranteed to localizes all nodes of the network except some boundary nodes and isolated weakly interior nodes.
  
  Our algorithm works under the condition that the strong interior of the network is connected. Relaxing this condition, it would be interesting to characterize the conditions under which localization starting from different components of the strong interior can be stitched together. It would be also interesting to study impact noisy distance measurement on our algorithm. Our algorithm also works under the strong assumption of uniform communication range. An important direction of future research would be to see if our approach can be extend to networks with sensors having irregular communication range, e.g., quasi unit disk networks \cite{kuhn2008ad}. Another problem is to compare the class of networks that are fully localized by our algorithm to those that are fully localized by trilateration. The example in Fig. \ref{honey} shows  a class of network in which trilateration does not progress beyond the base step for any choice of the initial triangle, but our algorithm always localizes all the nodes from any initial strongly interior node.

   \paragraph{Acknowledgements.} The first author is supported by NBHM, DAE, Govt. of India and the third author is supported by CSIR, Govt. of India. This work was done when the second author was at Jadavpur University, Kolkata, India, supported by UGC, Govt. of India. We would like to thank the anonymous reviewers for their valuable comments which helped us to improve the quality and presentation of the paper.

\bibliographystyle{plain}
\bibliography{wheel}

\appendix
\newpage
%  
%   \section{Omitted Proofs of Section \ref{constr}}\label{app1}
% 
% 
%  \subsection{Proof of Lemma \ref{perimeter_2cover} }
%   
%   
% 
%  

%    \subsection{Proof of Theorem \ref{communication_wheel_from_non_eclipsed} }

%  \subsection{Proof of Theorem \ref{maintheorem2} }

  \section{Appendix: Simulation Results}\label{compare}

\begin{figure}[thb!]
\centering
\subcaptionbox[Short Subcaption]{
        \label{configuration2_fig_a}
}
[
    0.48\textwidth % width of caption
]
{
    
    \includegraphics[width=.48\linewidth]{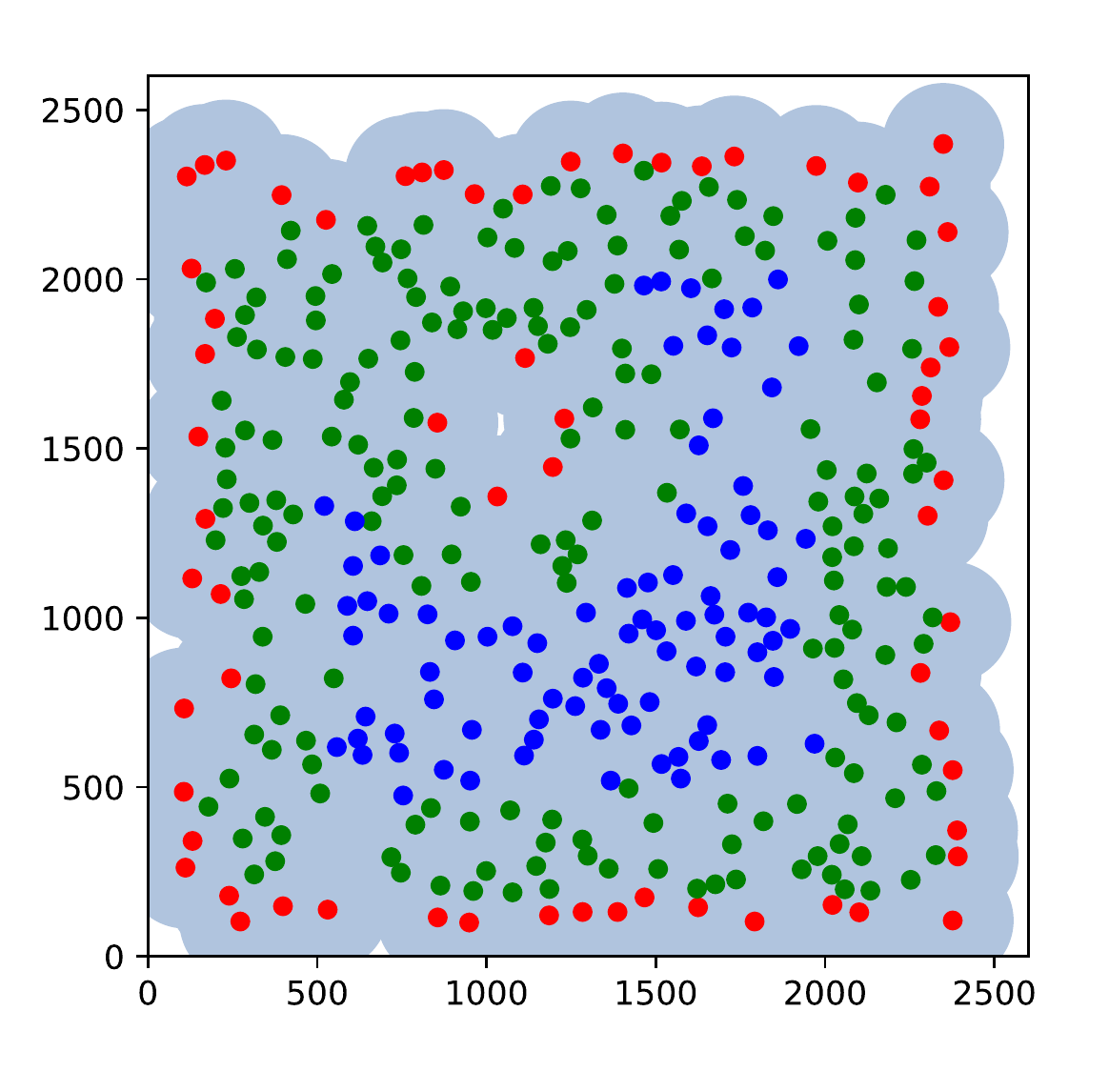}
}
\hfill % seperation
\subcaptionbox[Short Subcaption]{
    \label{configuration2_fig_b}
}
[
    0.48\textwidth % width of caption
]
{
    \includegraphics[width=.48\linewidth]{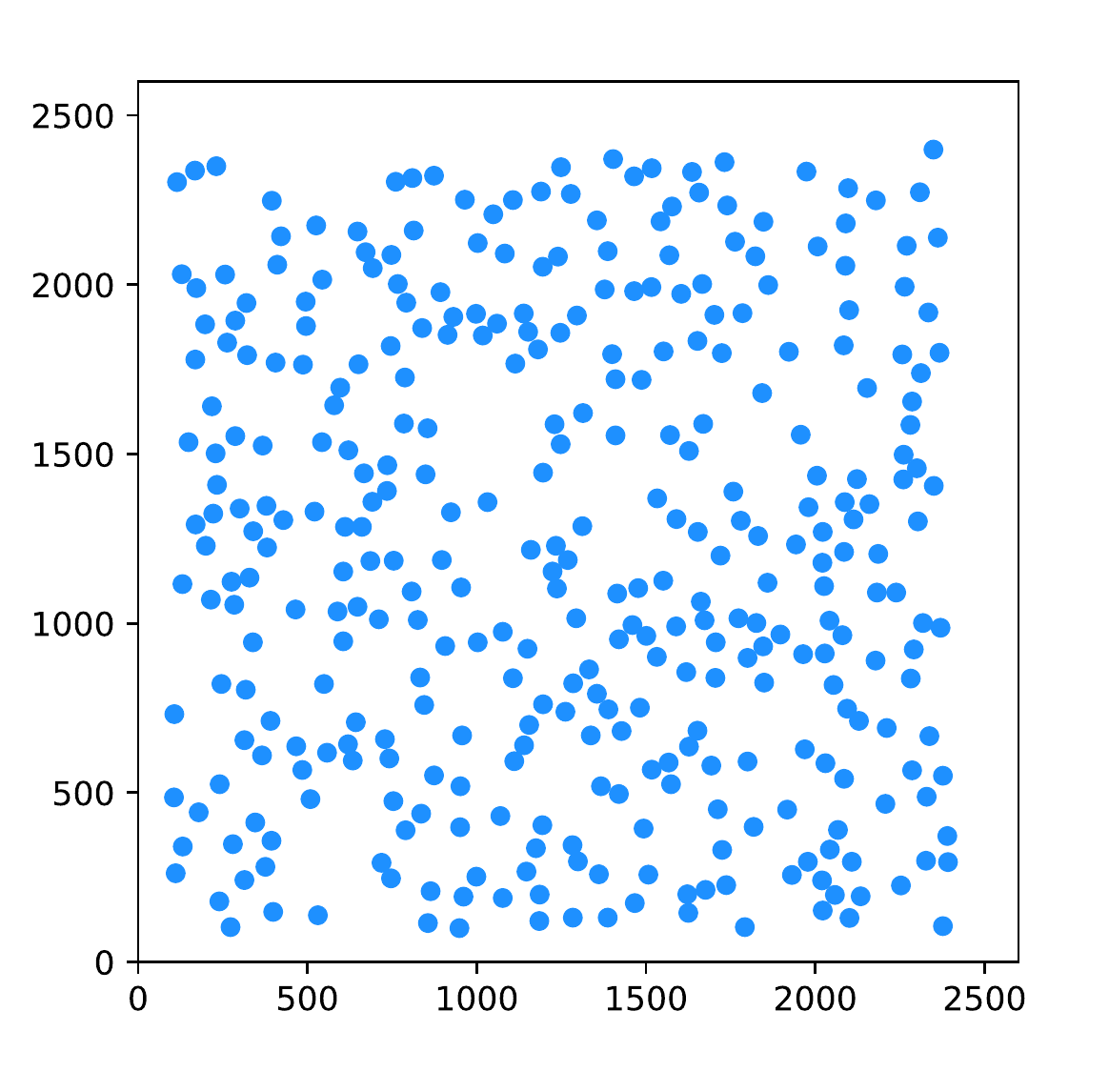}

}
\\
\subcaptionbox[Short Subcaption]{
       \label{sparse}
}
[
    0.48\textwidth % width of caption
]
{
    \includegraphics[width=.48\linewidth]{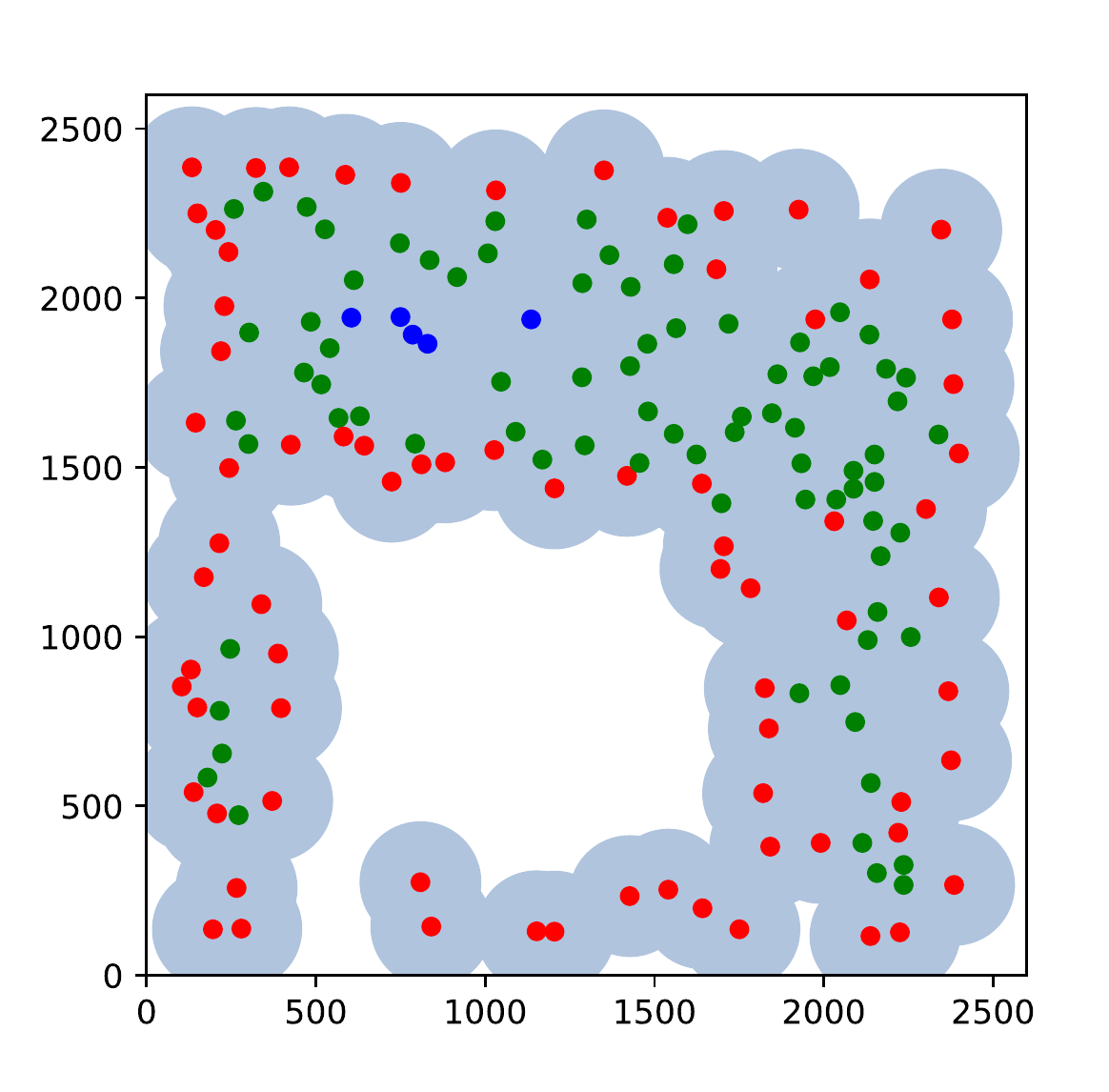}

}
\hfill
\subcaptionbox[Short Subcaption]{
       \label{configuration2_fig_d}
}
[
    0.48\textwidth % width of caption
]
{
   \includegraphics[width=.48\linewidth]{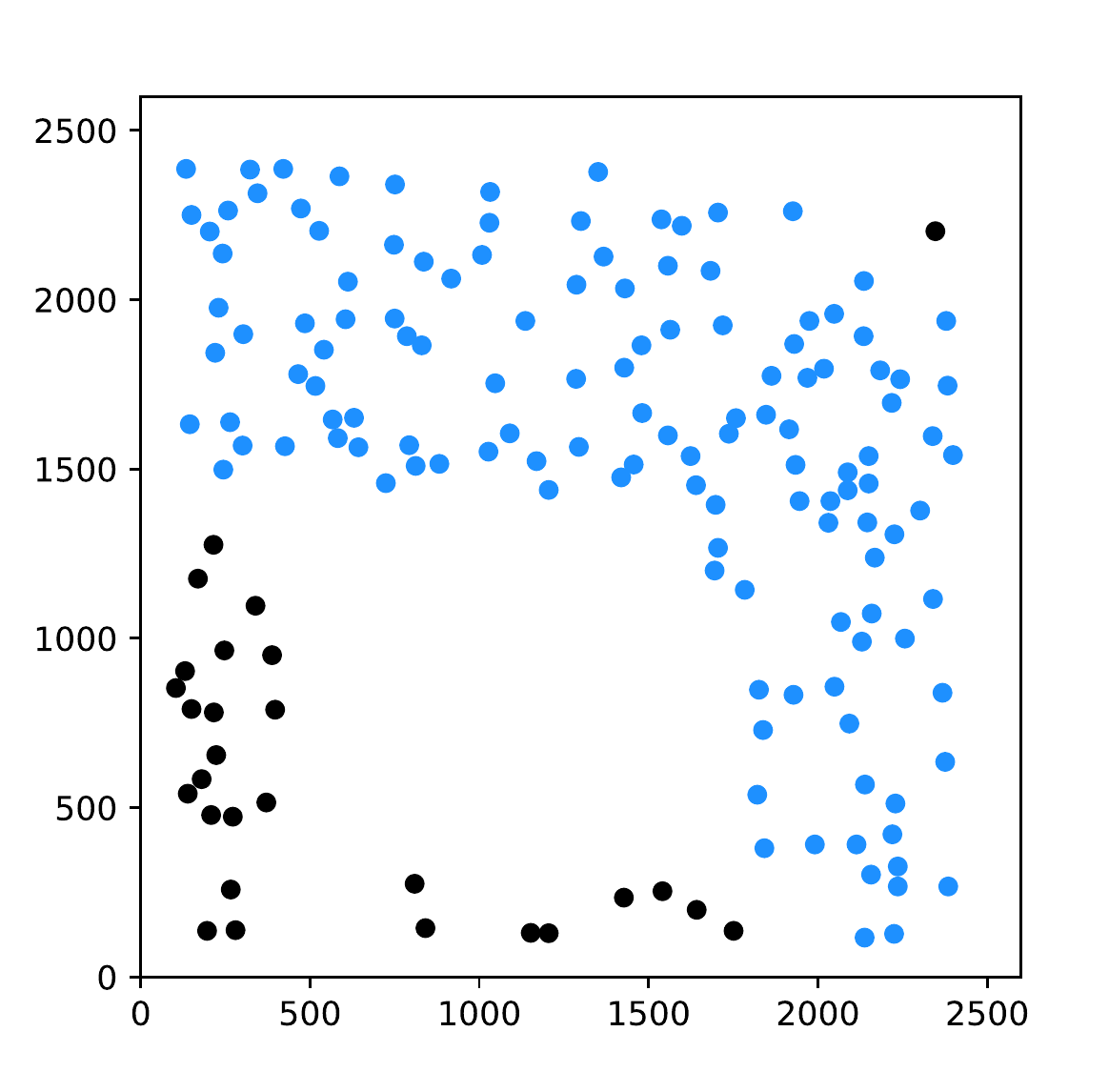}
}

\caption[Short Caption]{a) A random network of 350 sensor nodes. The red, green and deep blue nodes are respectively boundary, weakly interior and strongly interior nodes. Disks of radii $\frac{r}{2}$ around each node is shown in pale blue. b) The network being dense, all 350 nodes are successfully localized by our algorithm. c) A relatively sparse and irregular network of 160 sensor nodes. d) Despite having only 5 strongly interior nodes, our algorithm successfully localizes 132 nodes, and thus achieves the maximum number of nodes localized by trilateration starting from any possible triangle (See Fig.\ref{best}). The blue nodes are successfully localized, and the black nodes are not localized.}
\label{test}
\end{figure}

\begin{figure}[thb!]
\centering
\subcaptionbox[Short Subcaption]{
       \label{best}
}
[
    0.48\textwidth % width of caption
]
{
    
    \includegraphics[width=.48\linewidth]{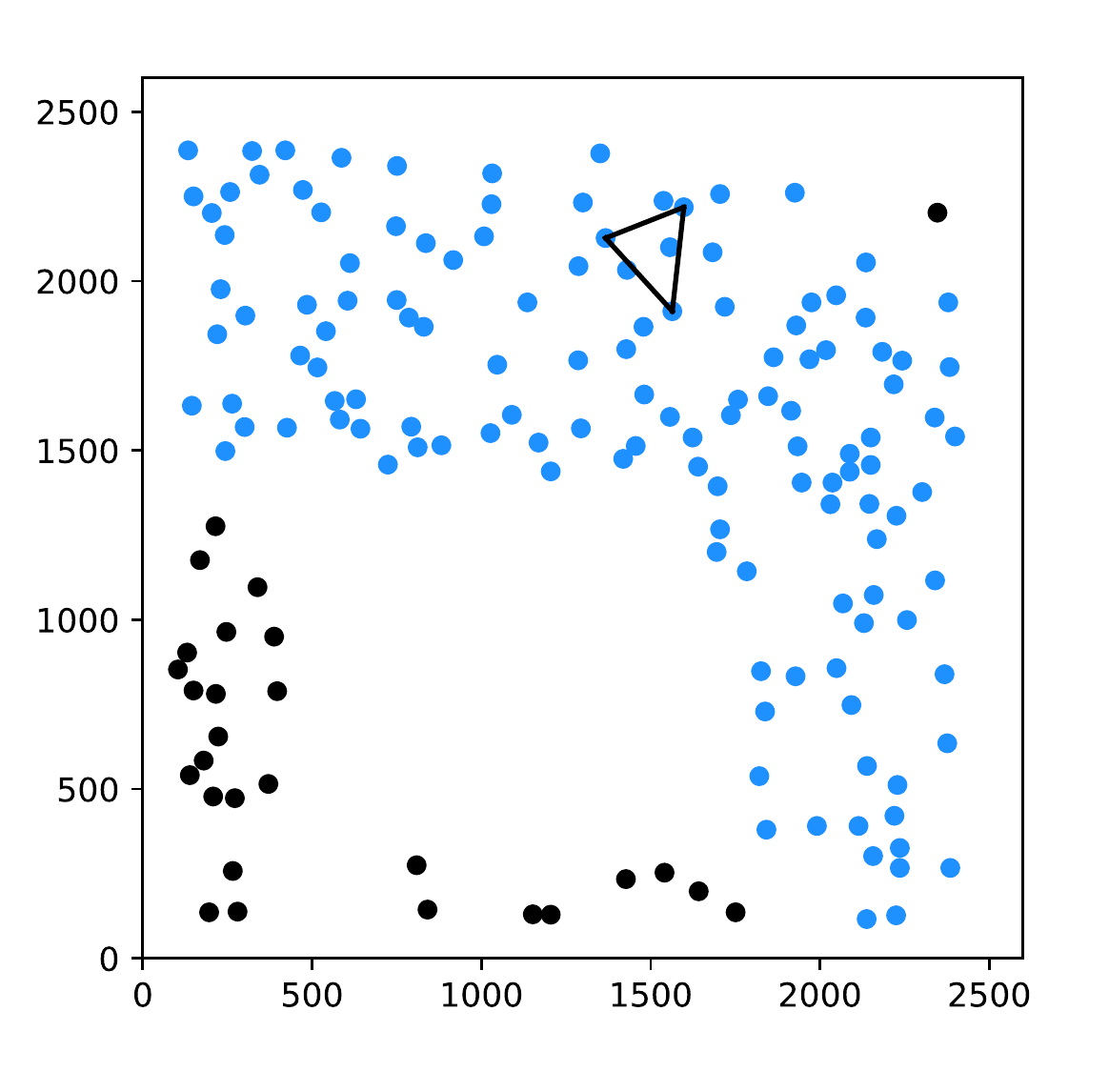}
}
\hfill % seperation
\subcaptionbox[Short Subcaption]{
    \label{configuration2_fig_b}
}
[
    0.48\textwidth % width of caption
]
{
    \includegraphics[width=.48\linewidth]{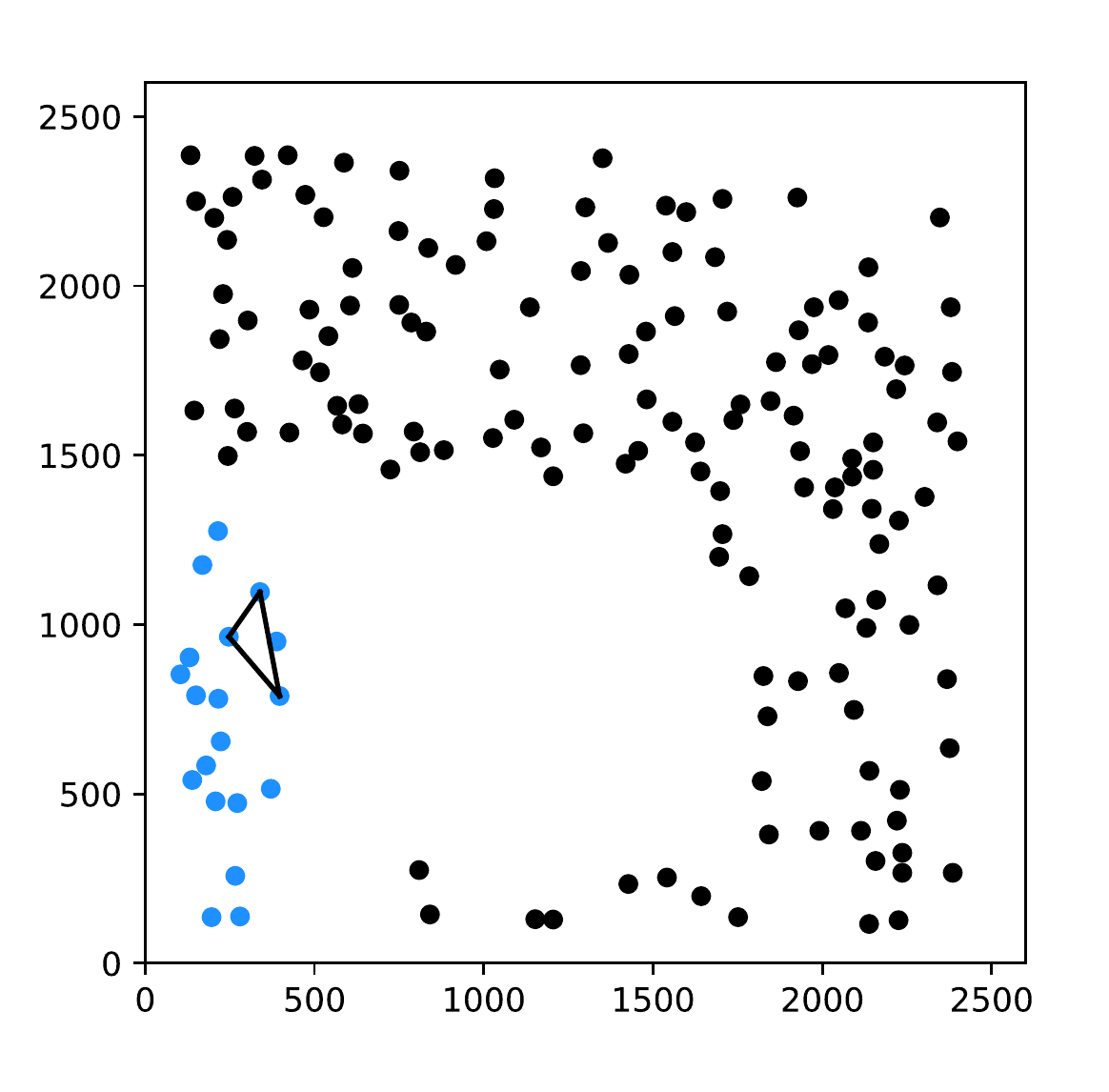}

}
\\
\subcaptionbox[Short Subcaption]{
       \label{stop}
}
[
    0.48\textwidth % width of caption
]
{
    \includegraphics[width=.48\linewidth]{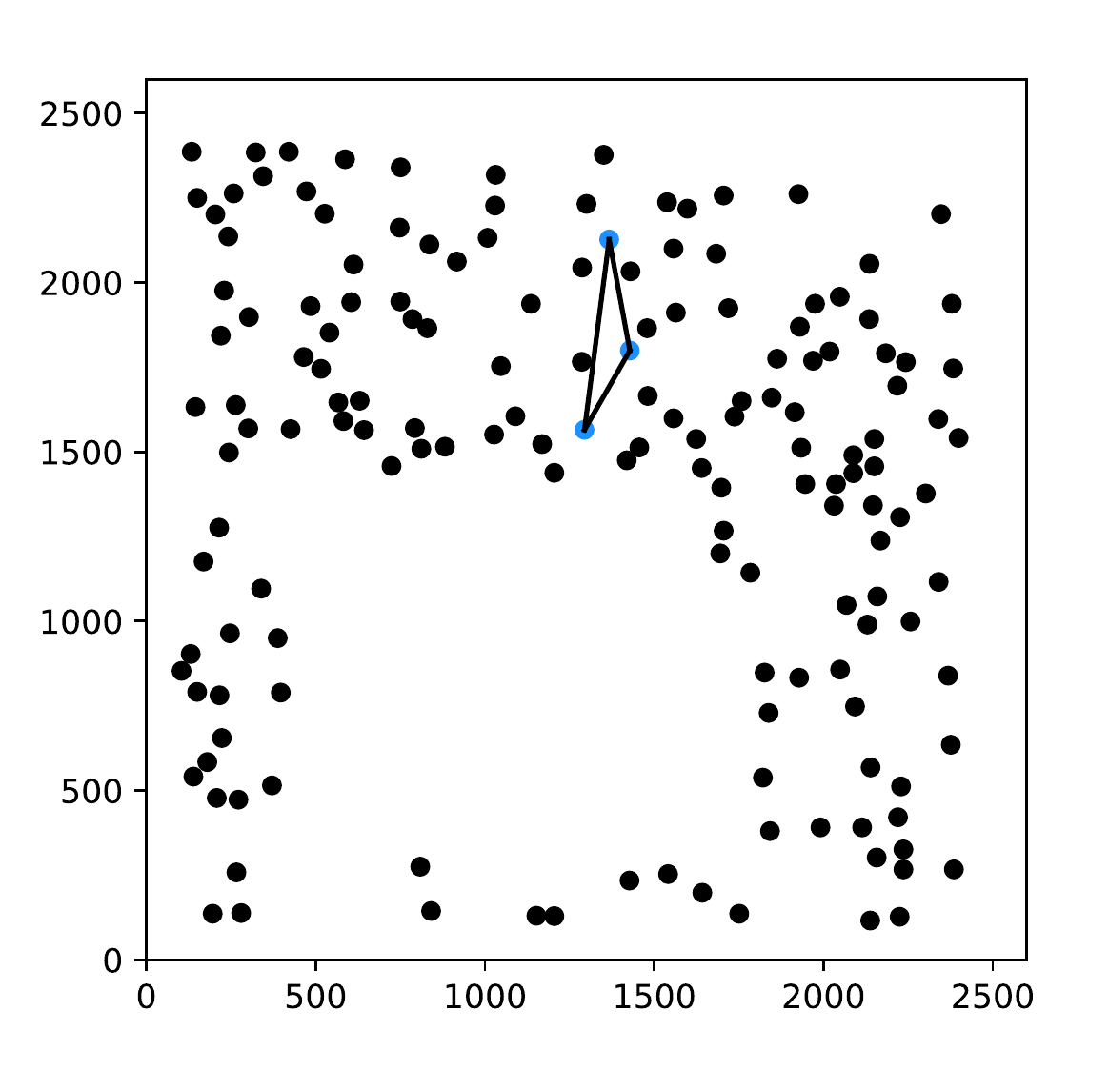}

}
\hfill
\subcaptionbox[Short Subcaption]{
        \label{stop2}
}
[
    0.48\textwidth % width of caption
]
{
   \includegraphics[width=.48\linewidth]{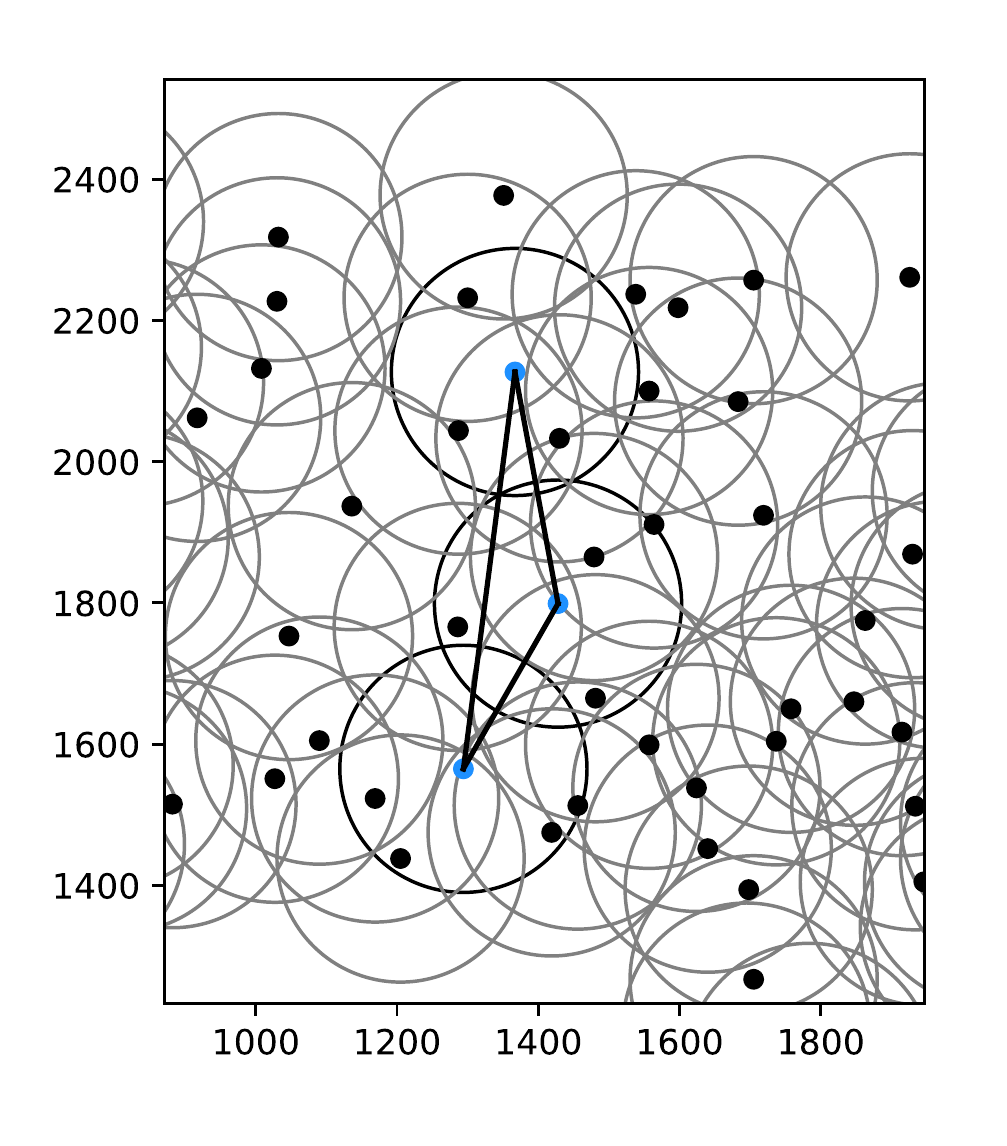}
}

\caption[Short Caption]{The same network from Fig. \ref{sparse} is consider and trilateration is attempted from different initial triangles. a) In the overall best result, 132 nodes are localized b) In this case, only 19 nodes are localized. c) In the worst case, trilateration does not progress beyond the initial triangle. d) A closer view of Fig. \ref{stop}. }
\label{test2}
\end{figure}

\end{document}